\documentclass{iacrcc}

\license{CC-by}

\title[running = {Computing endomorphism rings from one endomorphism}]{The supersingular endomorphism ring problem given one endomorphism}

\addauthor[inst = {1,2},
	orcid = {0009-0007-6116-6863},
	surname = {Herlédan Le Merdy},
	email = {arthur.herledan_le_merdy@ens-lyon.fr}
	]{Arthur Herlédan Le Merdy}

\addauthor[inst = {1},
	orcid = {0000-0003-1249-6077}
	]{Benjamin Wesolowski}

\addaffiliation[ror = {05n21n105},
	country = {France},
	city = {Lyon}
	]{ENS de Lyon, CNRS, UMPA, UMR 5669}

\addaffiliation[
	ror = {04msnz457},
	country = {France},
	city = {Lyon}
	]{ENS de Lyon, LIP (CNRS, U. Lyon, ENS de Lyon, Inria, UCBL), UMR 5668}

\addfunding[
	grantid = {ANR-21-CE94-0003},
	country = {France}
	]{CHARM ANR-NSF}

\addfunding[
	grantid = {ANR-22-PNCQ-0002},
	country = {France}
	]{HQI}

\addfunding[
	grantid = {ANR-22-PETQ-0008},
	country = {France}
	]{PEPR quantique France 2030}


\usepackage[english]{babel}
\usepackage[utf8]{inputenc}

\usepackage{cite}

\usepackage{amsfonts}

\usepackage{mathrsfs}

\usepackage{algorithm}
\usepackage[noend]{algpseudocode}

\usepackage{stmaryrd}

\usepackage{amsmath}

\usepackage{upgreek}

\usepackage{amssymb}

\usepackage{scalerel}
\usepackage{stackengine,wasysym}

\usepackage{xcolor}
\usepackage{todonotes}

\usepackage{xcolor}

\newcommand\reallywidetilde[1]{\ThisStyle{%
  \setbox0=\hbox{$\SavedStyle#1$}%
  \stackengine{-.1\LMpt}{$\SavedStyle#1$}{%
    \stretchto{\scaleto{\SavedStyle\mkern.2mu\AC}{.5150\wd0}}{.6\ht0}%
  }{O}{c}{F}{T}{S}%
}}

\usepackage{yhmath}

\usepackage{bm}

\usepackage{tikz-cd}

\DeclareMathOperator {\disc}{disc}
\DeclareMathOperator {\Cl}{Cl}
\DeclareMathOperator {\Unif}{Unif}
\DeclareMathOperator {\Cay}{Cay}
\DeclareMathOperator {\End}{End}
\DeclareMathOperator {\Iso}{Iso}
\DeclareMathOperator {\Aut}{Aut}
\DeclareMathOperator {\rank}{rank}

\DeclareMathOperator {\enc}{\texttt{enc}}
\DeclareMathOperator {\domain}{domain}
\DeclareMathOperator {\codomain}{codomain}
\DeclareMathOperator {\Sp}{Sp}

\newcommand{\mathsc}[1]{{\normalfont\textsc{#1}}}
\newcommand{\mathname}[1]{{\mathsc{#1}}}
\newcommand{\IsogenyPath}{\ensuremath{\ell\mathname{-IsogenyPath}}}
\newcommand{\EndRing}{\ensuremath{\mathname{EndRing}}}
\newcommand{\Vectorisation}{\ensuremath{\mathname{Vectorisation}}}
\newcommand{\Primitivisation}{\ensuremath{\mathname{Primitivisation}}}
\newcommand{\Effective}[1]{\ensuremath{\mathname{Effective}\ #1\mathname{-Vectorisation}}}
\newcommand{\HiddenShift}{\ensuremath{\mathname{Hidden Shift}}}

\date{}

\begin{document}

\maketitle

\begin{abstract}
Given a supersingular elliptic curve $E$ and a non-scalar endomorphism $\alpha$ of $E$, we prove that the endomorphism ring of $E$ can be computed in classical time about $|\disc(\mathbb Z[\alpha])|^{1/4}$, and in quantum subexponential time, assuming the generalised Riemann hypothesis. Previous results either had higher complexities, or relied on heuristic assumptions.

Along the way, we describe and analyse a general algorithm to divide isogenies in polynomial time, and to solve the Primitivisation problem in polynomial time.
Following the attacks on SIDH, isogenies in high dimension are a central ingredient of our results.
\end{abstract}

\begin{textabstract}
Given a supersingular elliptic curve E and a non-scalar endomorphism α of E, we prove that the endomorphism ring of E can be computed in classical time about disc(Z[α])^1/4, and in quantum subexponential time, assuming the generalised Riemann hypothesis. Previous results either had higher complexities, or relied on heuristic assumptions.

Along the way, we describe and analyse a general algorithm to divide isogenies in polynomial time, and to solve the Primitivisation problem in polynomial time.
Following the attacks on SIDH, isogenies in high dimension are a central ingredient of our results.
\end{textabstract}

\begin{section}{Introduction}
\noindent Isogeny-based cryptography is an active and promising branch of post-quantum cryptography.
Isogenies are certain kinds of maps between elliptic curves.
The security of cryptosystems in this family relies mainly on the algorithmic hardness of constructing an isogeny between two supersingular elliptic curves: the \emph{supersingular isogeny path problem}.

Endomorphisms of an elliptic curve $E$ are isogenies from $E$ to itself, and their collection forms the endomorphism ring $\End(E)$. 
The \emph{endomorphism ring problem}, denoted  \EndRing, consists in computing the endomorphism ring of a supersingular elliptic curve.
Under the generalised Riemann hypothesis, the isogeny problem is equivalent to \EndRing~\cite{EC:EHLMP18,wesolowski_supersingular_2022}.
This equivalence has placed \EndRing\ at the heart of isogeny-based cryptography,
and its hardness has been proved to relate to the security of the \texttt{CGL} hash function \cite{charles_cryptographic_2009, EC:EHLMP18}, the \texttt{CSIDH} key exchange protocol \cite{AC:CLMPR18, EC:CasPanVer20, EC:Wesolowski22} and the \texttt{SQISign} signature scheme \cite{AC:DKLPW20}.

In certain cryptosystems, the elliptic curves involved are equipped with one public endomorphism. For instance, in \texttt{CSIDH}~\cite{AC:CLMPR18}, all elliptic curves are defined over $\mathbb F_p$, and thus the Frobenius endomorphism is an accessible non-trivial endomorphism. The situation is similar in \cite{chenu_higher-degree_2022,PKC:DFKLMP23}.
The endomorphism ring problem then asks to find all the \emph{other} endomorphisms.
This yields the following question:
\begin{itemize}
	\item How much does knowing one endomorphism simplify the computation of the endomorphism ring of a supersingular elliptic curve? 
\end{itemize}
A closely related question was studied in~\cite{arpin_orienteering_2023}: given two curves $E$ and $E'$, together with two endomorphisms $\alpha \in \End(E)$ and $\beta \in \End(E')$, how hard is it to find an isogeny between them?
Under several heuristic assumptions, they provide a classical exponential algorithm and a quantum subexponential algorithm solving this problem.
With the equivalence between the isogeny path problem and \EndRing, their work provides a first answer to the above question. Yet that answer has limitations: first, as stated, it is only heuristic. Second, the output of the algorithm of~\cite{arpin_orienteering_2023} may have exponential size, which could considerably increase the cost of applying the equivalence.\\

The schemes of \cite{AC:CLMPR18,chenu_higher-degree_2022,PKC:DFKLMP23} have in common the notion of \emph{orientation}, introduced by Colò and Kohel \cite{colo_orienting_2020}. 
Given an order $\mathfrak O$ in a quadratic number field, an $\mathfrak O$-orientation of a curve $E$ is a subring of $\End(E)$ isomorphic to $\mathfrak O$. 
The interest in this notion lies in the fact that the set of $\mathfrak O$-oriented curves comes with an action of the class group of $\mathfrak O$.
The problem of inverting this group action is known as the \Vectorisation\ problem. The presumed hardness of this problem 
was already at the heart of the security of the \texttt{CRS} protocol \cite{couveignes_hard_2006,rostovtsev_public-key_2006} where the action of ideal class groups came from the complex multiplication theory of ordinary elliptic curves.
Today, it is behind the security of \texttt{CSIDH}~\cite{AC:CLMPR18} and its variants~\cite{chenu_higher-degree_2022,PKC:DFKLMP23}.

Any endomorphism $\alpha \in \End(E) \setminus \mathbb Z$ gives rise to a $\mathbb Z[\alpha]$-orientation, hinting at the connection between the \EndRing\ problem given one endomorphism, and problems involving orientations.
The link between \Vectorisation\ and \EndRing\ has first been studied in the particular case of \texttt{CSIDH} in \cite{EC:CasPanVer20}.
That article proves that there is a subexponential-time reduction from breaking \texttt{CSIDH} to computing the endomorphism rings.
Then it has been improved and extended to a polynomial-time equivalence between \Vectorisation\ and \EndRing\ in \cite{EC:Wesolowski22}.
However, these results necessitate the orientations to be \emph{primitive}: the quadratic suborder must be maximal in the endomorphism ring. Obtaining a primitive orientation from a given orientation is not trivial \textemdash\ this problem is called the \Primitivisation\ problem.
This leads us to this second question:

\begin{itemize}
	\item How hard is it to get a primitive orientation from an orientation?
\end{itemize}

The \Primitivisation\ problem was first introduced in \cite{arpin_orienteering_2023} as a presumably hard problem, and they gave a quantum subexponential algorithm solving it. 

\begin{subsection}{Orientations and variants of \EndRing}

We now give an informal overview of orientations and related hard problems.
For formal definitions, we refer the reader to Section \ref{sec_def} about orientations and general notations, to Section \ref{sec_endring} about the different hard problems and to \cite{silverman_arithmetic_1986} for a detailed reference about elliptic curves and isogenies.

We fix a prime integer $p$ and we denote $E$ a supersingular elliptic curve defined over $\bar{\mathbb{F}}_p$.
An isogeny of elliptic curves is a morphism between elliptic curves seen as abelian varieties.
We denote by $\End(E)$ the ring formed by isogenies from $E$ to itself, i.e. the endomorphisms of $E$.
We consider the following supposedly hard problem \EndRing. 

\begin{itemize}
	\item \EndRing: Given a supersingular elliptic curve $E$, compute $\End(E)$.
\end{itemize}

The current best classical algorithms to solve \EndRing\ run in expected time $\tilde{O}(p^{1/2})$, see for instance \cite{eisentrager_computing_2020}, and the best quantum algorithms have complexity in $\tilde{O}(p^{1/4})$, see for example \cite{INDOCRYPT:BiaJaoSan14}.\\

Let $\mathfrak{O}$ be an order of a quadratic number field.
An orientation $\iota$ is an embedding from $\mathfrak{O}$ into $\End(E)$.
This is mainly equivalent to knowing an endomorphism in $\End(E)\backslash\mathbb{Z}$.
If this embedding cannot be extended to any superorder of $\mathfrak{O}$, we say that $\iota$ is a primitive orientation.
When a (primitive) orientation $\iota$ exists, we say that $E$ is (primitively) $\mathfrak{O}$-orientable and that the pair $(E,\iota)$ is a (primitively) $\mathfrak{O}$-oriented elliptic curve.\\

This notion of orientation comes together with variants of \EndRing\ where partial information on the endomorphism ring is given.
Let $\alpha$ be an element of the quadratic order $\mathfrak{O}$.

\begin{itemize}
	\item $\alpha$-\EndRing: Given a supersingular elliptic curve $E$ together with an orientation $\iota : \mathbb{Z}[\alpha] \hookrightarrow \End(E)$, compute $\End(E)$.
\end{itemize}

\begin{itemize}
	\item $\mathfrak{O}$-\EndRing: Given a primitively $\mathfrak{O}$-oriented supersingular elliptic curve $(E,\iota)$, compute $\End(E)$.
\end{itemize}

These two problems are tightly related.
On the one hand, there is a direct reduction from $\mathfrak{O}$-\EndRing\ to $\alpha$-\EndRing\ as the inputs of the former are also inputs of the latter.
On the other hand, the reduction from $\alpha$-\EndRing\ to $\mathfrak{O}$-\EndRing\ is not trivial as it requires to compute a primitive orientation from any given orientation.
This computation has been introduced in \cite{arpin_orienteering_2023} as a hard problem together with a quantum algorithm for solving it in subexponential time under some heuristics.
\begin{itemize}
	\item \Primitivisation: Given a supersingular elliptic curve $E$ together with an orientation $\iota : \mathbb{Z}[\alpha] \hookrightarrow \End(E)$, find the primitive orientation $\iota': \mathfrak{O} \hookrightarrow \End(E)$ such that $\mathbb{Z}[\alpha] \subseteq \mathfrak{O}$. 
\end{itemize}
To get a classical reduction from $\alpha$-\EndRing\ to $\mathfrak{O}$-\EndRing, the most direct approach consists in solving the \Primitivisation\ problem to extend the given orientation to a primitive one.\\

The $\mathfrak{O}$-\EndRing\ problem is not only interesting to investigate the complexity of \EndRing\ given some additional information, it also has an important place in isogeny-based cryptography.
To see that, we first need to consider the \Vectorisation\ problem induced by primitive orientations.
From a primitive orientation, one can construct a free action of the class group $\Cl(\mathfrak{O})$ over the set of primitively $\mathfrak{O}$-oriented elliptic curves.
We denote this group action as
\begin{align*}
	\Cl(\mathfrak{O}) \times SS_\mathfrak{O}(p) &\rightarrow SS_\mathfrak{O}(p)\\
	([\mathfrak{a}], (E,\iota)) &\mapsto \mathfrak{a} \star (E,\iota) := (E^\mathfrak{a}, \iota^\mathfrak{a}),
\end{align*}
where $SS_\mathfrak{O}(p)$ is the set of primitively $\mathfrak{O}$-oriented supersingular elliptic curves defined over $\bar{\mathbb{F}}_p$ up to isomorphism.
This group action allows one to define a \Vectorisation\ problem, giving a framework to study security of \texttt{CSIDH}-like protocols. 
\begin{itemize}
	\item $\mathfrak{O}$-\Vectorisation: Given $(E,\iota)$ and $(E',\iota')$ in $SS_\mathfrak{O}(p)$ find an $\mathfrak{O}$-ideal $\mathfrak{a}$ such that $E^\mathfrak{a} \simeq E'$.
\end{itemize}

Under the generalised Riemann hypothesis and given the factorisation of $\disc(\mathfrak{O})$, the $\mathfrak{O}$-\EndRing\ problem is equivalent to $\mathfrak{O}$-\Vectorisation\  in probabilistic polynomial time, see \cite{EC:Wesolowski22}.
Therefore, the security of many protocols such as \texttt{CSIDH} \cite{AC:CLMPR18}, \texttt{CSI-FiSh} \cite{AC:BeuKleVer19} and \texttt{CSURF} \cite{PQCRYPTO:CasDec20} reduces to $\mathfrak{O}$-\EndRing\, see \cite{EC:Wesolowski22}.

In the current state of the art, using $l$ to denote the length of the input, the problem of $\mathfrak{O}$-\Vectorisation\ can heuristically be solved in expected classical time $l^{O(1)}|\disc(\mathfrak{O})|^{1/4}$, using for instance approaches close to the ones in \cite{DCC:DelGal16}.
Quantumly, it can heuristically be solved in time $l^{O(1)} L_{|\disc(\mathfrak{O})}[1/2]$, see \cite[Proposition 4]{EC:Wesolowski22}.
\end{subsection}

\begin{subsection}{Contributions}
We provide algorithms whose asymptotic complexity matches or improves upon previous results in the literature. Unlike previous results, our proofs do not rely on heuristic assumptions.
In this list of contributions, we suppose that the input and output of the algorithms are always in efficient representation, we refer the reader to Section \ref{subsec_encoding} for more information about representations and encodings.
\begin{itemize}

\item In Section \ref{sec_higher}, we develop the first ingredient for the rest of the paper: an algorithm to divide isogenies, Corollary~\ref{cor_division_general}.
Explicitly, given two isogenies $\varphi$ and $\eta$, the algorithm returns the unique isogeny $\psi$ such that $\varphi = \psi \circ \eta$ (or asserts that such $\psi$ does not exist). This is the right-division of $\varphi$ by $\eta$ (dualizing, one can also divide on the left).

The heart of the method is an algorithm to divide isogenies by integers in polynomial time.
It is a generalization of a division algorithm for translates of the Frobenius introduced by Robert to compute the endomorphism ring of an ordinary elliptic curve, \cite[Theorem 4.2]{robert_applications_2022}.
Before the attacks of SIDH, dividing isogenies by integers either required superpolynomial time, or degraded the quality of the representation (getting exponentially worse with each application).
It is not the case here.

\item In Section \ref{sec_prim}, we use this division algorithm to solve the \Primitivisation\ problem.
This result adapts Robert's algorithm for computing in polynomial time the endomorphism ring of ordinary elliptic curves \cite{robert_applications_2022}, which can be seen as an ordinary counterpart of the \Primitivisation\ problem.
We prove that, when the factorisation of $\disc(\mathfrak{O})$ is known, there is a classical polynomial time algorithm solving \Primitivisation.

As an application, we provide a polynomial time algorithm for computing the action of smooth ideals.
Previous polynomial-time algorithms for this task required the norm of the input ideal to be powersmooth.

\end{itemize}
We now use $l$ to denote the length of the input, and use the standard $L$-notation for subexponential complexities (Definition~\ref{def:L-notation}).
\begin{itemize}
	\item In Section \ref{sec_vect}, under the generalised Riemann hypothesis, we provide
	\begin{itemize} 
		\item a classical algorithm solving $\mathfrak{O}$-\Vectorisation\ in $l^{O(1)} |\disc(\mathfrak{O})|^{1/4}$ expected time.
		\item a quantum algorithm solving $\mathfrak{O}$-\Vectorisation\ in $l^{O(1)} L_{|\disc(\mathfrak{O})|}[1/2]$ expected time.
	\end{itemize}
	This directly leads to
	\begin{itemize}
		\item a classical algorithm solving $\mathfrak{O}$-\EndRing\ in $l^{O(1)} |\disc(\mathfrak{O})|^{1/4}$  expected time.
		\item a quantum algorithm solving $\mathfrak{O}$-\EndRing\ in $l^{O(1)} L_{|\disc(\mathfrak{O})|}[1/2]$ expected time.
	\end{itemize}
Combined with our resolution of \Primitivisation, we obtain the following theorems on solving the endomorphism ring problem knowing an endomorphism, rigorously.
\begin{theorem}[\textbf{GRH}]
\label{the_classical}
There is a classical algorithm that given a supersingular curve $E$, and an endomorphism $\alpha \in \End(E)\setminus \mathbb Z$, computes the endomorphism ring of $E$ in expected time $l^{O(1)} |\disc(\mathbb{Z}[\alpha])|^{1/4}$ where $l$ is the length of the input.
\end{theorem}
\begin{theorem}[\textbf{GRH}]
\label{the_quantum}
There is a quantum algorithm that given a supersingular curve $E$, and an endomorphism $\alpha \in \End(E)\setminus \mathbb Z$, computes the endomorphism ring of $E$ in expected time $l^{O(1)}L_{|\disc(\mathbb{Z}[\alpha])|}[1/2]$ where $l$ is the length of the input.
\end{theorem}

\item In section \ref{sec_navigate}, we detail how the algorithmic improvements of Section \ref{sec_higher} allow one to navigate efficiently in the volcano of oriented isogenies.
In the previous literature, the number of steps that one could efficiently take in a volcano was limited because of the degrading quality of representations.

As a direct application, we present an optimisation of the resolution of $\mathfrak{O}$-\EndRing\ through the following reduction:
	\begin{itemize}
		\item Under the generalised Riemann hypothesis, there is a probabilistic reduction from $(\mathbb{Z} + c \mathfrak{O})$-\EndRing\ to $\mathfrak{O}$-\EndRing\ taking a time polynomial in the length of the input and in the largest prime factor of $c$.
	\end{itemize}
This last result improves the probabilistic polynomial reduction given by \cite[Theorem 5]{EC:Wesolowski22} by relaxing the powersmoothness constraint on $c$.
It also leads to a classical algorithm solving $(\mathbb{Z} + c\mathfrak{O})$-\EndRing\ in expected time polynomial in the length of the input, in $\disc(\mathfrak{O})$ and in the largest prime factor of $c$.
	This improves and removes the heuristics of \cite[Corollary 6.]{EC:Wesolowski22}.
\end{itemize}

\end{subsection}

\ifcsstring{@IACRversion}{submission}{}{
\subsection*{Acknowledgements}
The authors would like to extend their gratitude to Guillaume Hanrot for helpful discussions and feedback which have significantly contributed to the writing of this paper.
They also wish to thank Damien Robert for his invaluable answers to our questions. 
The authors were supported by the CHARM ANR-NSF grant (ANR-21-CE94-0003), by the HQI initiative (ANR-22-PNCQ-0002) and by the PEPR quantique France 2030 programme (ANR-22-PETQ-0008).}

\end{section}
\begin{section}{Definitions and notations}
\label{sec_def}
In this article, some results will be proved assuming the generalised Riemann hypothesis. 
They will be marked by \textbf{GRH}, see for instance Theorem \ref{the_classical} and \ref{the_quantum}.

We denote by $\mathbb{F}_q$ the finite field with $q$ elements and by $\bar{\mathbb{F}}_q$ its algebraic closure.
The cardinality of a set $S$ is denoted by $\#S$.
For any order $\mathcal{O}$, $\disc(\mathcal{O})$ is the notation of its discriminant.	
We use the standard $O$-notation together with the $\tilde{O}$-notation which removes the logarithmic factors of the $O$-notation, i.e.
$O(f(x) \log^k(x)) = \tilde{O}(f(x))$
for any positive integer $k$.
When a time complexity is given without specifying the type of operations, it is assumed to count binary operations.

In the algorithms presented here, we use $x \gets \Unif\{S\}$ to denote that $x$ is sampled uniformly at random among the elements of $S$.

\begin{definition}[(Power)Smoothness bound]
Let $n$ be an integer of prime decomposition $\ell_1^{e_1} \dots \ell_r^{e_r}$.
We say that an integer $B$ is a \textbf{smoothness bound} on $n$ and $n$ is said to be $\bm{B}$\textbf{-smooth} if 
$$B \geq \max\limits_{i \in \llbracket 1,r \rrbracket}\ell_i;$$
if further 
$$B \geq \max\limits_{i \in \llbracket 1,r \rrbracket} \ell_i^{e_i}$$
then $B$ is a \textbf{powersmoothness bound} on $n$ and $n$ is $\bm{B}$\textbf{-powersmooth}.
We denote by $P^+(n)$ the integer $\max \limits_{i \in \llbracket 1,r \rrbracket } \ell_i$ and by $P^*(n)$ the integer $\max\limits_{i \in \llbracket 1,r \rrbracket} \ell_i^{e_i}$.
\end{definition}

\begin{definition}[Extension degree]
For any elliptic curve $E$ defined over a finite field $\mathbb{F}_{p^k}$ and integer $n$ of prime decomposition $\ell_1^{e_1}\dots\ell_r^{e_r}$, we use the following notations
\begin{itemize}
	\item $\delta_{E}(n) := \max\limits_{i \in \llbracket 1,r \rrbracket}[\mathbb{F}_{p^k}(E[\ell_i^{e_i}]) : \mathbb{F}_{p^k}]$,
	\item $\delta_{E,2}(n) := \max\limits_{(i,j) \in \llbracket 1,r \rrbracket^2, i \neq j} [\mathbb{F}_{p^k}(E[\ell_i^{e_i} \ell_j^{e_j}]):\mathbb{F}_{p^k}]$,
\end{itemize}
where, for any integer $m$, $\mathbb{F}_{p^k}(E[m])$ stands for the smallest field extension of $\mathbb{F}_{p^k}$ where the coordinates of the points of $E[m]$ live.
\end{definition}

\begin{definition}[$L$-notation]\label{def:L-notation}
Let $a,b,x$ be three real numbers.
To handle subexponential complexities, we define the following standard $\bm{L}$\textbf{-notation}
$$
L_x[a,b] := \exp(b(\log x)^a(\log \log x)^{(1-a)}),
$$
as well as this $L$-notation for unknown constants
$$
L_x[a] := \exp(O((\log x)^a(\log \log x)^{(1-a)})).
$$
\end{definition}

\begin{subsection}{Cayley graph}
\label{subsec_graph}
\begin{definition}[Cayley graph]
Let $G$ be a finite group and let $S \subseteq G$ be a generating subset of $G$.
The \textbf{Cayley graph} $\bm{\Cay(G,S)}$ is the graph whose vertices are the elements of $G$ and such that there exists an edge between two vertices $g_1,g_2$ if and only there exists an $s \in S$ such that $g_2 = sg_1$.
\end{definition}

We shall use the following result of Childs, Jao and Soukharev regarding random walks over Cayley graphs of class groups.

\begin{proposition}[\textbf{(GRH}) Theorem 2.1 in \cite{childs_constructing_2014}]
\label{pro_walk}
Let $\mathfrak{O}$ be an imaginary quadratic order of discriminant $\Delta$ and conductor $f_\mathfrak{O}$.
Let $\varepsilon > 0$ and $x$ be a real number such that  $x \geq (\log|\Delta|)^{2 + \varepsilon}$.
Let $\Sigma_x$ be the set 
$$\{ [\mathfrak{p}] \in \Cl(\mathfrak{O}) \text{ such that } \gcd(f_\mathfrak{O},\mathfrak{p}) = 1 \text{ and } N(\mathfrak{p}) \leq x \text{ prime}\}$$
from which we define the set $S_x$ to be
$$\Sigma_x \cup \{ [\mathfrak{p}]^{-1} \text{ for }  [\mathfrak{p}] \in \Sigma_x \}.$$
Then there exists a positive constant $C > 1$, depending only on $\varepsilon$, such that for all $\Delta$ sufficiently large, a random walk of length
$$
t \geq C \frac{\log \#\Cl(\mathfrak{O})}{\log \log |\Delta|}
$$
in the Cayley graph $\Cay(\Cl(\mathfrak{O}), S_x)$ from any starting vertex lands in any fixed subset $H \subset \Cl(\mathfrak{O})$ with probability $P$ such that 
$$
 \frac{1}{2}\frac{\#H}{\#\Cl(\mathfrak{O})} \leq  P.
$$
\end{proposition}

\end{subsection}

\begin{subsection}{Elliptic curves and orientations}
\label{subsect_orientations}
In this section, we recall some basic definitions and notations about elliptic curves before introducing the recent notions of orientations \cite{colo_orienting_2020}.
For more details about elliptic curves theory, we refer the reader to Silverman's book \cite{silverman_arithmetic_1986}.\\

An \textbf{elliptic curve} is an abelian variety of dimension 1.
\textbf{Isogenies} of elliptic curves are non-trivial homomorphisms between them.
An isogeny from an elliptic curve to itself is called an \textbf{endomorphism}.
The set of all endomorphisms of an elliptic curve $E$ together with the trivial map form the \textbf{endomorphism ring} $(\End(E),+,\circ)$ where $+$ is the point-wise addition and $\circ$ is the composition of maps.
For any integer $n$ and elliptic curve $E$, we denote by $[n]$ the \textbf{multiplication-by-$\bm{n}$ map} on $E$ and by $E[n]$ its kernel, called the $\bm{n}$\textbf{-torsion subgroup} of $E$.
An elliptic curve $E$ defined over a finite field of characteristic $p$ is said to be \textbf{supersingular} if $E[p] \simeq \{ 0 \} $.

In this paper, we only work with supersingular elliptic curves defined over a field of characteristic $p$, where $p$ is a fixed prime number.
The $p^k$-Frobenius isogeny from a curve $E$ is the isogeny $\phi^{p^k}_E : E \to E^{(p^k)} : (x,y) \mapsto (x^{(p^k)},y^{(p^k)})$, where $E^{(p^k)}$ is defined by the equation of $E$ with coefficients raised to the power $p^k$.
An isogeny $\varphi : E \to E'$ is \textbf{inseparable} if it factors as $\varphi = \psi \circ \phi^{p}_E$ for some isogeny $\psi$. Otherwise, the isogeny is \textbf{separable}.
Any isogeny can be written as $\varphi = \psi \circ \phi^{p^k}_E$ where $\psi$ is separable, and we define the \textbf{degree} of $\varphi$ as $\deg(\varphi) = p^k\cdot\#(\ker \varphi)$.
For any isogeny $\varphi : E \to E'$, there exists a unique isogeny $\hat \varphi$, the \textbf{dual isogeny} of $\varphi$, such that $\hat \varphi \circ \varphi = [\deg(\varphi)]$.

For any prime $\ell \not = p$, an \textbf{$\bm{\ell}$-isogeny} is an isogeny of degree $\ell$.\\

An important property of supersingular elliptic curves is that their endomorphism ring is isomorphic to a maximal order of the quaternion algebra over $\mathbb{Q}$ ramified only at $p$ and at infinity.
This quaternion algebra is unique up to isomorphism and we denote it by $B_{p,\infty}$.
More explicitly, we have the isomorphism of $\mathbb{Q}$-algebras
$$ B_{p,\infty} \simeq \mathbb{Q} + \mathbb{Q}i + \mathbb{Q}j + \mathbb{Q}ij \text{ such that } i^2 = -p, j^2 =-q_p \text{ and } ij = -ji$$
where $q_p$ is a positive integer depending only on $p$.
We refer the reader to \cite{voight_quaternion_2021} for more information about quaternion algebras.

In a way, quaternion algebras can be seen as two imaginary quadratic number fields combined to get a non commutative 4-dimensional $\mathbb{Q}$-algebra.
It is in fact possible to embed an infinite number of imaginary quadratic number fields into a given $\mathbb{Q}$-algebra $B_{p,\infty}$.
Naturally, one can then study how orders of imaginary quadratic number fields embed into a given endomorphism ring of supersingular elliptic curves.
The study of such embeddings in isogeny-based cryptography originates from \cite{colo_orienting_2020}, where they are introduced as \textbf{orientations}.

Let $K$ be an imaginary quadratic number field such that $p$ does not split in $K$.

\begin{definition}[Orientation]
An elliptic curve $E$ is said to be $\bm{K}$\textbf{-orientable} if there exists an embedding $\iota: K \hookrightarrow \End(E)\otimes_\mathbb{Z}\mathbb{Q}$. Then the embedding $\iota$ is a $\bm{K}$\textbf{-orientation} and the pair $(E,\iota)$ is a $\bm{K}$\textbf{-oriented} elliptic curve.

For any order $\mathcal{O}$ of $K$, we say that $\iota$ is an $\bm{\mathcal{O}}$\textbf{-orientation} and $(E,\iota)$ is an $\bm{\mathcal{O}}$\textbf{-oriented} elliptic curve if $\iota(\mathcal{O}) \subseteq \End(E)$. In this case, $\iota$ will often be considered as the embedding $\mathcal{O} \hookrightarrow \End(E)$.
An $\mathcal{O}$-orientation $\iota$ is \textbf{primitive} if $\iota(\mathcal{O}) = \End(E)\cap\iota(K)$.
In this case, the oriented elliptic curve $(E,\iota)$ is said to be \textbf{primitively $\bm{\mathcal{O}}$-oriented}.
\end{definition}

\begin{definition}[Oriented isogeny]
Let $(E,\iota)$ and $(E',\iota')$ be two $K$-oriented elliptic curves, and let $\varphi : E \rightarrow E'$ be an isogeny.
We denote by $\varphi_*(\iota)$ the $K$-orientation induced by $\varphi$.
Explicitly, this orientation is given by
$$
\varphi_*(\iota)(\kappa) = (\varphi \circ \iota(\kappa) \circ \hat\varphi) \otimes \frac{1}{\deg(\varphi)}, \forall \kappa \in K.
$$
We say that $\varphi$ is \textbf{$\bm K$-oriented} if $\varphi_*(\iota)$, is equal to $\iota'$.
In particular, a $K$-oriented isogeny of degree $1$ is called a $K$-oriented isomorphism.
\end{definition}

For any order $\mathfrak{O}$ in $K$, let $SS_\mathfrak{O}(p)$ be the set of primitively $\mathfrak{O}$-oriented supersingular elliptic curves over $\overline{\mathbb{F}}_p$, up to $K$-oriented isomorphism.
By \cite[Proposition 3.2]{onuki_oriented_2021}, if $p$ does not divide the conductor of $\mathfrak{O}$ then $SS_\mathfrak{O}(p)$ is not empty.
From now, we assume that this is always the case.
One can then define a free action of the class group $\Cl(\mathfrak{O})$ on the set of curves $SS_\mathfrak{O}(p)$.
This is analogous to the well-known action of $\Cl(\mathfrak{O})$ on the set of ordinary elliptic curves whose endomorphism rings are isomorphic to $\mathfrak{O}$.

Let us describe precisely how $\Cl(\mathfrak{O})$ acts on $SS_\mathfrak{O}(p)$.\\

We consider the action of an invertible $\mathfrak{O}$-ideal $\mathfrak{a}$ prime to $p$ on a primitively $\mathfrak{O}$-oriented elliptic curve $(E,\iota) \in SS_\mathfrak{O}(p)$.
First, we consider the finite subgroup $E[\mathfrak{a}]$ of $E$, called the \textbf{$\mathfrak{a}$-torsion} of $E$, given by
$$
E[\mathfrak{a}] := \bigcap_{\alpha \in \mathfrak{a}} \ker \iota(\alpha).
$$
It induces a separable isogeny $\varphi_\mathfrak{a} : E \rightarrow E/E[\mathfrak{a}]$ of kernel $E[\mathfrak{a}]$.
We call this isogeny $\varphi_\mathfrak{a}$ the \textbf{$\mathfrak{a}$-multiplication} and its image curve $E/E[\mathfrak{a}]$, also denoted $E^\mathfrak{a}$, the \textbf{$\mathfrak{a}$-transform}.
Then the action of $\mathfrak{a}$ on $(E,\iota)$ is the primitively $\mathfrak{O}$-oriented supersingular elliptic curve $(E^\mathfrak{a},(\varphi_\mathfrak{a})_*(\iota))$ up to $K$-isomorphism.
By factorisation, we get the whole action of $\Cl(\mathfrak{O})$ on $SS_\mathfrak{O}(p)$.

\begin{proposition}[\hspace{-0.45pt}\cite{onuki_oriented_2021}]
\label{pro_groupaction}
The class group $\Cl(\mathfrak{O})$ acts over $SS_\mathfrak{O}(p)$ freely and has at most two orbits.
We denote this action as
\begin{align*}
	\Cl(\mathfrak{O}) \times SS_\mathfrak{O}(p) &\rightarrow SS_\mathfrak{O}(p) \\
	([\mathfrak{a}],(E,\iota)) &\mapsto \mathfrak{a} \star (E,\iota) := (E^\mathfrak{a}, (\varphi_\mathfrak{a})_*(\iota)).
\end{align*}
In addition, for any given orbit $O$ and any given primitively $\mathfrak{O}$-oriented supersingular elliptic curve $(E,\iota)$, either $(E,\iota)$ or its $\mathfrak{O}$-twist $(E,\bar{\iota})$, where $\bar{\iota}(\alpha) = \iota(\bar{\alpha})$, is in $O$.
\end{proposition}
\end{subsection}

\begin{proof}
This proposition is obtained from \cite[Theorem 3.4]{onuki_oriented_2021} and from \cite[Proposition 3.3]{onuki_oriented_2021}.
In particular, inside the proof of \cite[Proposition 3.3]{onuki_oriented_2021}, it is shown that either $(E,\iota)$ or $(E,\bar{\iota})$ is in a given orbit. 
\end{proof}

\begin{subsection}{Representation of isogenies}
\label{subsec_encoding}

\begin{definition}[Efficient representation, following {\cite[Definition 1.3]{wesolowski_random_2024}}]\label{def:efficient-representation}
Let $\mathscr A$ be a polynomial time algorithm. It is an \textbf{efficient isogeny evaluator} if for any $D \in \{0,1\}^*$ such that $\mathscr A(\mathtt{validity}, D)$ outputs~$\top$,
there exists an isogeny $\varphi: E \to E'$ (defined over some finite field $\mathbb F_q$) such that:
\begin{enumerate}
\item on input $(\mathtt{curves}, D)$, $\mathscr A$ returns $(E,E')$,
\item on input $(\mathtt{degree}, D)$, $\mathscr A$ returns $\deg(\varphi)$,
\item on input $(\mathtt{eval},D,P)$ with $P\in E(\mathbb F_{q^k})$, $\mathscr A$ returns $\varphi(P)$.
\end{enumerate}
If furthermore $D$ is of polynomial size in $\log(\deg \varphi)$ and $\log q$, then $D$ is an \textbf{efficient representation} of $\varphi$ (with respect to $\mathscr A$).
\end{definition}

There are several ways of representing efficiently isogenies.
For instance, by considering an isogeny $\varphi$ as a chain of smaller degree isogenies $\varphi_1 \circ \dots \circ \varphi_n$, one can represent it by the list of kernels $\ker(\varphi_i)_{i=1}^n$ and use Vélu's formulae as the main procedure of the efficient isogeny evaluator.
When the degree $\varphi$ is smooth, it then is possible to find a suitable chain of isogenies to represent $\varphi$ efficiently with this method.
This is an important example as it was one of the most used representation in isogeny-based cryptography before SIDH's attack.
Thanks to higher dimensional isogenies, we now have access to interpolation algorithms without any smoothness constraints.
One can for instance use \cite{EC:Robert23}, or directly the results of Section~\ref{sec_higher}, to efficiently represent an isogeny from its image of a large enough subgroup.

For a more exhaustive list of efficient ways to represent isogenies, we refer the reader to \cite{robert_efficient_2024}.
One can consider a single general efficient isogeny evaluator which, depending of some bits of the input, uses a different subalgorithm (higher dimensional interpolation, evaluation of the isogenies as a chain, etc.). Therefore, in the rest of the article, the algorithm $\mathscr A$ is left implicit, and we simply say that an isogeny $\varphi$ is in efficient representation.

\begin{definition}[Efficient representation of orientation]
Let $(E,\iota)$ be an $\mathfrak{O}$-orientated elliptic curve.
An \textbf{efficient representation} of $\iota$ is a pair $(\omega, D)$ where $\omega$ is a generator of $\mathfrak{O}$, and $D$ is an efficient representation of $\iota(\omega) \in \End(E)$.
\end{definition}

We now define a function $\enc$, introduced in \cite{EC:Wesolowski22}, which returns a unique encoding of the $K$-isomorphism class of a $\mathfrak{O}$-oriented elliptic curve.
It takes as input a representation of an oriented elliptic curve $(E,\iota) \in SS_\mathfrak{O}(p)$ and returns a unique triple $(E,P,Q)$ assuming that we have fixed in advance:

\begin{itemize}
	\item A canonical representative for each $\overline{\mathbb F}_p$-isomorphism class of elliptic curves over $\overline{\mathbb F}_p$;
		for instance, the curve of equation $E: y^2 + xy = x^3 - (36x + 1)/(j(E) - 1728)$ for any $j(E) \not \in \{ 0,1728 \}$, see \cite[page 52]{silverman_arithmetic_1986}, 
	\item A generator $\omega$ of $\mathfrak{O}$, typically one with the smallest possible norm,
	\item A deterministic procedure that takes as input an elliptic curve $E$ in canonical form and returns a point $P \in E$ of order greater than $4N(\omega)$.
\end{itemize}

Then the map $\enc: (E,\iota) \mapsto (E,P,Q)$ is given by constructing the point $P$ of order greater than $4 N(\omega)$ using the deterministic procedure and setting $Q$ to be $\iota(\omega)(P)$.
As shown in \cite{EC:Wesolowski22}, this encoding is an unique encoding of the $K$-isomorphism class of $(E,\iota)$.
Moreover, when $\iota$ is efficiently represented, the encoding $\enc(E,\iota)$ can be computed in polynomial time.
Thus checking if two $\mathfrak{O}$-oriented elliptic curves are $K$-isomorphic is done in polynomial time using $\enc$.\\

When the $j$-invariant of the oriented curve is $0$ or $1728$, one needs to use another canonical form, see \cite[page 52]{silverman_arithmetic_1986}, for instance
$$
	E: y^2 + y = x^3 \text{, if } j = 0
$$
and 
$$
	E: y^2 = x^3 + x \text{, if } j = 1728.
$$
In these cases, one also needs to consider the non-trivial automorphisms of the elliptic curve and thus to replace $Q$ by the set $\{ (\sigma_* \iota) (\omega) (P) | \sigma \in \Aut(E) \}$.

Finally, we define the image of any set $S$ of oriented supersingular elliptic curves by $\enc$ as the set of their unique encoding by $\enc$, denoted $\enc(S)$.

This unique encoding of $K$-isomorphism class of $\mathfrak{O}$-oriented elliptic curves is a crucial ingredient for Section \ref{sec_vect}:
First, it provides an efficient method for checking whether two $\mathfrak{O}$-oriented elliptic curves belong to the same class, which is central to finding collisions in the meet-in-the-middle approach for the classical resolution of the $\mathfrak{O}$-\Vectorisation\ problem.
Second, to properly use the Kuperberg's algorithm and quantumly solve $\mathfrak{O}$-\Vectorisation, it is essential to have a unique representation of the oriented elliptic curves we are dealing with, i.e. up to $K$-isomorphism.

\end{subsection}

In this paper, unless otherwise specified, when an algorithm takes as input an isogeny, we mean that the isogeny is given with an efficient representation.
It is also the case for orientations taken as input.

\end{section}
\begin{section}{The endomorphism ring problem and its friends}
\label{sec_endring}
One of the central problems in (supersingular) isogeny-based cryptography is the following \IsogenyPath\ problem, where $\ell$ is a prime number.

\begin{problem}[\IsogenyPath]
Given two supersingular elliptic curves $E$ and $E'$ over $\mathbb{F}_{p^2}$ and a prime $\ell \not = p$, find a chain of $\ell$-isogenies from $E$ to $E'$.
\end{problem}

The \IsogenyPath\ problem is considered to be the fundamental problem at the heart of isogeny-based cryptography. 
This problem has been shown to be equivalent, first under some heuristics \cite{EC:EHLMP18} then only under GRH \cite{wesolowski_supersingular_2022}, to the problem of finding the structure of the endomorphism ring of a supersingular elliptic curve.
This second problem is called the endomorphism ring problem; here we refer to it as the \EndRing\ problem.
Since for any supersingular elliptic curve $E$ defined over a finite field of characteristic $p$, $\End(E)$ is isomorphic to a maximal order of the quaternion algebra $B_{p,\infty}$, the \EndRing\ problem comes in two flavors.
One can either look for four isogenies generating $\End(E)$ as a lattice or for four quaternions generating a maximal order which is isomorphic to $\End(E)$.
The notion of $\varepsilon$-basis unifies those approaches under GRH, see \cite{wesolowski_supersingular_2022}.

\begin{definition}[$\varepsilon$-basis]
Let $\varepsilon : B_{p,\infty} \rightarrow \End(E)\otimes_\mathbb{Z} \mathbb{Q}$ be an isomorphism and $L \subseteq B_{p,\infty}$ be a lattice.
We call a pair $(\alpha,\theta)$, where $(\alpha_i)^{\rank{L}}_{i=1}$ is a basis of $L$ and $\theta_i = \varepsilon(\alpha_i)$, an $\varepsilon$-basis of $L$.
The pair $(\alpha,\theta)$ will be called an $\varepsilon$-basis of $\varepsilon(L)$ as well.
\end{definition}

\begin{problem}[\EndRing]
Given a supersingular elliptic curve $E$ over $\mathbb{F}_{p^2}$, find an $\varepsilon$-basis of $\End(E)$.
\end{problem}

The current best classical algorithms to solve \EndRing\ run in expected time $\tilde{O}(p^{1/2})$, see for instance \cite{eisentrager_computing_2020} (or~\cite[Theorem~8.8]{EC:PagWes24} for a heuristic-free algorithm), and the best quantum algorithms have complexity in $\tilde{O}(p^{1/4})$, see for example \cite{INDOCRYPT:BiaJaoSan14}.

The notion of orientation (see Section \ref{subsect_orientations}), introduced by Colò and Kohel in \cite{colo_orienting_2020}, induces a variant of \EndRing\ where partial information on the endomorphism ring is given.

\begin{problem}[$\mathfrak{O}$-\EndRing]
Given a primitively $\mathfrak{O}$-oriented supersingular elliptic curve $(E,\iota)$ over a finite field of characteristic $p$, find an $\varepsilon$-basis of $\End(E)$.
\end{problem}

The study of this problem is not only important to see how the complexity of $\EndRing$ is impacted by the knowledge of a single non-trivial endomorphism but also because it is in fact equivalent, under GRH, to the $\mathfrak{O}$-\Vectorisation\ problem \cite{EC:Wesolowski22}.

\begin{problem}[$\mathfrak{O}$-\Vectorisation]
Given $(E,\iota),(E',\iota') \in SS_\mathfrak{O}(p)$ two oriented supersingular elliptic curves, find an $\mathfrak{O}$-ideal $\mathfrak{a}$ such that $E^\mathfrak{a} \simeq E'$.
\end{problem}

This hardness of the problem $\mathfrak{O}$-\Vectorisation\ measures whether or not an action induced by a primitive orientation is a one-way function.
Hence, recovering the keys of \texttt{CSIDH}-like protocols, such as \cite{AC:CLMPR18,chenu_higher-degree_2022,PKC:DFKLMP23}, reduces to $\mathfrak{O}$-\Vectorisation.\\

In this paper, we only need the following reduction between the two problems.

\begin{proposition}[\textbf{GRH}, Proposition 7 in \cite{EC:Wesolowski22}]
\label{pro_OEndRingreduction}
Given the factorisation of $\disc(\mathfrak{O})$, the $\mathfrak{O}$-\EndRing\ problem reduces to $\mathfrak{O}$-\Vectorisation\  in probabilistic polynomial time in the length of the instance.
\end{proposition}

In the current state of the art, the $\mathfrak{O}$-\Vectorisation\  problem can heuristically be solved in expected classical time $l^{O(1)}|\disc(\mathfrak{O})|^{1/4}$, with $l$ the length of the input, using for instance a meet-in-middle approach in a similar way as presented in \cite{DCC:DelGal16}.
Quantumly, $\mathfrak{O}$-\Vectorisation\  can heuristically be solved in subexponential time in the length of the discriminant of $\mathfrak{O}$, see \cite[Proposition 4]{EC:Wesolowski22}.
Hence, knowing an orientation seems to make a significant difference in the expected runtime to solve \EndRing.

However, those resolutions and reductions need the orientation to be primitive and assume several heuristics.
When the orientation is not primitive, this is equivalent to knowing one non-trivial endomorphism of the curve.
Obviously, knowing an orientation also gives the knowledge of a non-trivial endomorphism.
In the other direction, given a non-trivial endomorphism, one can compute in polynomial time its degree and trace \cite[Proposition 81]{kohel_endomorphism_1996} and deduce a quadratic number $\alpha$ such that $\mathbb{Z}[\alpha] \hookrightarrow \End(E)$ is an orientation.\\

We introduce a variant of $\mathfrak{O}$-\EndRing\ where the given orientation is not required to be primitive.

\begin{problem}[$\alpha$-\EndRing]
Given a supersingular elliptic curve $E$ over $\mathbb{F}_{p^2}$ and an orientation $\iota: \mathbb{Z}[\alpha] \hookrightarrow \End(E)$, find an $\varepsilon$-basis of $\End(E)$.
\end{problem}

The following \Primitivisation\ problem has been introduced in \cite{arpin_orienteering_2023} as a hard problem.
It forms a bridge between the $\alpha$-\EndRing\ and $\mathfrak{O}$-\EndRing\ problems.

\begin{problem}[\Primitivisation]
Given a supersingular elliptic curve $E$ over $\mathbb{F}_{p^2}$ and an orientation $\iota: \mathbb{Z}[\alpha] \hookrightarrow \End(E)$, find a primitive orientation $\iota' : \mathfrak{O} \hookrightarrow \End(E)$ such that the order $\mathbb{Z}[\alpha]$ is contained in the order $\mathfrak{O}$.
\end{problem}

In \cite{arpin_orienteering_2023} the authors also give a quantum algorithm for solving it in subexponential time under some heuristics, and conjecture that it is hard to solve in general.
Moreover, they give a quantum algorithm to solve the \IsogenyPath\ problem given non-primitive orientation that uses their Primitivisation algorithm as a subprocedure.
Inevitably, their algorithm inherits the need for heuristics of the subprocedure.
In Section \ref{sec_higher}, we develop tools involving higher dimensional isogenies.
These tools are used in Section \ref{sec_prim} in an algorithm solving \Primitivisation\ for an orientation $\iota: \mathbb{Z}[\alpha] \hookrightarrow \End(E)$ in classical polynomial time given the factorisation of $\disc(\mathbb{Z}[\alpha])$.
It directly yields a classical subexponential and a quantum polynomial reduction of $\alpha$-\EndRing\ to $\mathfrak{O}$-\EndRing.

This implies, together with Proposition \ref{pro_OEndRingreduction}, reductions of the $\alpha$-\EndRing\ problem to the $\mathfrak{O}$-\Vectorisation\ problem.
This reasoning is formalized in Section \ref{sec_vect} together with a rigorous analysis of the complexity of $\mathfrak{O}$-\Vectorisation.
\end{section}

\begin{section}{Efficient division of isogenies}
\label{sec_higher}

In this section, we discuss higher dimensional isogenies and how they can be used to efficiently divide isogenies by integers.
The goal of this section is to prove Theorem~\ref{cor_division} below (and its more precise formulation Theorem~\ref{the_division}).
\begin{theorem}
\label{cor_division}
Algorithm \ref{alg_division} takes as input
\begin{itemize}
	\item Two elliptic curves $E_1$ and $E_2$ defined over $\mathbb{F}_{p^k}$,
	\item An isogeny $\varphi: E_1 \rightarrow E_2$ over $\mathbb{F}_{p^k}$ in efficient representation,
	\item An integer $n < \deg \varphi$,
\end{itemize}
and returns an efficient representation of $\varphi/n$ if this quotient is an isogeny (and otherwise returns \texttt{False}), and runs in time polynomial in $k \log p$ and $\log \deg(\varphi)$.

More precisely, this returned representation of $\varphi/n$ is of size $O(k \log(p) \log^3(\deg \varphi))$ and 
allows one to evaluate it at any point in $\tilde{O}(\log^{11}(\deg \varphi))$ operations 
over its field of definition.
\end{theorem}

Before proving it, let us state an immediate corollary.
\begin{corollary}[General division of isogenies]
\label{cor_division_general}
There is an algorithm which on input
\begin{itemize}
	\item Three elliptic curves $E_1$, $E_2$ and $E_3$ defined over $\mathbb{F}_{p^k}$,
	\item Two isogenies $\varphi: E_1 \rightarrow E_2$, $\eta: E_1 \rightarrow E_3$ over $\mathbb{F}_{p^k}$ in efficient representation,
\end{itemize}
returns an efficient representation of the isogeny $\psi$ such that $\varphi = \psi \circ \eta$ if it exists (and otherwise returns \texttt{False}), and runs in time polynomial in the length of the input.
\end{corollary}

\begin{proof}
Apply Theorem~\ref{cor_division} to the isogeny $\tilde \varphi = \varphi \circ \hat \eta$ and the integer $n = \deg(\eta)$. The isogeny $\psi$, if it exists, is precisely $\tilde \varphi/n$.
\end{proof}

The machinery of higher dimensional isogenies is only used in this section, 
and the reader may admit the main theorem and skip the rest of the section without impairing global understanding.
Some definitions require notions of algebraic geometry which will not be recalled here.
If necessary, the reader may refer to \cite{milne_abelian_1986}.
Let us emphasise that the ideas underlying Theorem~\ref{cor_division} and its proof originate from \cite{robert_applications_2022}.
Theorem~\ref{cor_division} and its proof are simply expressed in higher generality and greater detail than \cite{robert_applications_2022} provides.
\\

Initially, the idea of exploiting higher dimensional isogenies for efficient computation of isogenies between elliptic curves was introduced by Castryck and Decru to attack SIDH, see \cite{EC:CasDec23}. 
Among other ingredients, this attack relies on the generalization of Vélu's formulae by Lubicz and Robert, see \cite{lubicz_computing_2012}.
The attack of Castrick and Decru has since been developed further, notably by Robert who generalised the attack \cite{EC:Robert23} and found other applications to elliptic curves, see \cite{robert_evaluating_2022} and \cite{robert_applications_2022}.\\

We now turn to a presentation of principally polarised abelian varieties \textemdash \: which constitute the suitable generalisation of elliptic curves to higher dimensions, see e.g. \cite{milne_abelian_1986}.

For any abelian varieties $A$ and $A'$ and any isogeny $\varphi : A \rightarrow A'$, the dual variety of $A$ is denoted by $\hat{A}$ and the dual isogeny of $\varphi$ is denoted by $\hat{\varphi} : \hat{A'} \rightarrow \hat{A}$.
\begin{definition}[(Principally) Polarised abelian varieties]
Let $A$ be an abelian variety.
One can derive from an ample divisor of $A$ an isogeny $\lambda : A \rightarrow \hat{A}$.
Such an isogeny is called a \textbf{polarisation} of $A$.
It is a \textbf{principal polarisation} if $\lambda$ is an isomorphism.
For any (principal) polarisation $\lambda$ of $A$, the pair $(A,\lambda)$ is a \textbf{(principally) polarised abelian variety}. 
\end{definition}

\begin{remark}
Elliptic curves are principally polarised abelian varieties having a unique principal polarisation, simply denoted by $\lambda_E$ for a given elliptic curve $E$.
Hence, there exists a natural principal polarisation over products of elliptic curves induced by the product of the polarisations.
We call this polarisation the \textbf{product polarisation} and denote it $\lambda_{E_1 \times \dots \times E_n}$ for the product of elliptic curve $E_1 \times \dots \times E_n$. 
When we consider a product of elliptic curves as a principally polarised abelian variety without specifying the polarisation, it means that it is polarised by the product polarisation.\\
\end{remark}

In this section, we shall mostly focus on isogenies between products of elliptic curves.

Let $E_1,\dots,E_n,E'_1,\dots,E'_m$ be elliptic curves and $\varphi_{i,j} : E_j \rightarrow E'_i$ be isogenies of elliptic curves where $i \in \llbracket 1,m \rrbracket, j \in \llbracket 1, n \rrbracket$.
From this set of isogenies, we naturally get the following map between products of elliptic curves
\begin{align*}
E_1\times \dots \times E_n &\longrightarrow E'_1 \times \dots \times E'_m\\
(P_1,\dots,P_n) &\longmapsto ( \sum_{j=1}^n \varphi_{1,j}(P_j), \dots ,\sum_{j=1}^n \varphi_{m,j}(P_j)).
\end{align*}
This map can be represented by the matrix ${(\varphi_{i,j})}_{i \in \llbracket 1,m \rrbracket, j \in \llbracket 1,n \rrbracket}$ called the \textbf{matrix form}.
If it has a finite kernel and $n=m$ then it is an isogeny.

Given an isogeny $\varphi : E \rightarrow E'$ between elliptic curves, one can construct an isogeny, denoted $\varphi^{\times n}$, in dimension $n$ from $E^n$ to ${E'}^n$ by setting
$${\varphi^{\times n}}(P) = (\varphi(P_1),\dots,\varphi(P_n)), \forall P = (P_1,\dots,P_n) \in E^n.$$

The matrix form of $\varphi^{\times n}$ is then the identity matrix of dimension $n$ ``multiplied'' by $\varphi$.
Any isogeny $F : E_1 \times \dots \times E_n \rightarrow E'_1 \times \dots \times E'_n$ between two products of elliptic curves can be written using a matrix form with the following injection maps 
\begin{center}
\begin{tabular}{r l l}
	$\tau_j:$& $E_j$ & $\longrightarrow E_1\times \dots \times E_n$\\
	& $P$ & $\longmapsto (\underbrace{0,\dots,0}_{j - 1},P,\underbrace{0,\dots,0}_{n-j})$
\end{tabular}
\end{center}
and projection maps
\begin{center}
\begin{tabular}{r c l} 
	$\pi_i:$& $E'_1 \times \dots \times E'_n $ & $\longrightarrow E'_i$\\
	& $(P_1,\dots,P_n)$ & $\longmapsto P_i$
\end{tabular}
\end{center}
with $i,j \in \llbracket 1,n \rrbracket $.
Indeed, by defining the isogeny $F_{i,j} : E_j \rightarrow E'_1$ as $\pi_i \circ F \circ \tau_j$, for all $i,j \in \llbracket 1,n \rrbracket$, we get
$$
F(P_1,\dots,P_n) = \Big(\sum_{j=1}^n F_{i,j}(P_j)\Big)_{1 \leq i \leq n},
$$
for any $(P_1,\dots,P_n) \in E_1 \times \dots \times E_n$.
We thus define the matrix form of the isogeny $F$ as $M(F) = (F_{i,j})_{i,j \in \llbracket 1,n \rrbracket}$.\\

\begin{remark}
\label{rem_theta}
In this paper, we shall consider only isogenies in higher dimensions whose domain is a product of elliptic curves that are principally polarised by the product polarisation.
However, to compute such isogenies, one needs to use a coordinate system capable of handling any principally polarised abelian varieties of the same dimension.
Currently, the most commonly used coordinate system is the theta model.
In particular, this is the case in the generalisation of the Vélus' formulae we shall use.

The theta coordinates of a product of elliptic curves can be obtained by multiplying the theta coordinates of each elliptic curve in the product.
One can compute theta coordinates of an elliptic curve directly from its $4$-torsion subgroup.
Therefore, converting a principally polarised product of elliptic curves to the theta model is inexpensive.
We shall not delve deeper into the machinery of theta functions here; further information can be found in \cite{EC:DLRW24} and \cite{robert_efficient_2021}.
\end{remark}

Thanks to the previous notations, we can provide a formal definition of what we mean by embedding an isogeny in higher dimensions.
For convenience, we generalise the notion of efficient representation of isogenies between elliptic curves to isogenies between products of elliptic curves.
Mainly, we require an efficient representation of isogenies between elliptic curves to be associated to an algorithm such that one can compute the image of any tuple of points in time polynomial in the length of the representation and in the size of the field over which those points are defined.

\begin{definition}[Embedding representation]
Let $n$ be an integer and $E_i, E'_i$ be elliptic curves for $i \in \llbracket 1,n \rrbracket$. 
Let $\varphi : E \rightarrow E'$ be an isogeny such that $E \in \{E_1,\dots,E_n\}$ and $E' \in \{E'_1,\dots,E'_n\}$.
An \textbf{embedding representation} of $\varphi$ in dimension $n$ is a triplet $(F,i,j)$ associated to a representation of $F$, where $F : E_1 \times \dots \times E_n \rightarrow E_1' \times \dots \times E'_n$, $i,j \in \llbracket 1,n \rrbracket$ and such that $\varphi(P) = \pi_j \circ F \circ \tau_i$ for any $P \in E$.
\end{definition}

We now introduce a notion of duality with respect to the principal polarisations allowing us to define a notion of isogenies between principally polarised abelian varieties behaving in a very similar way to elliptic curve isogenies.

\begin{definition}[$N$-Isogenies]
Let $(A,\lambda)$ and $(A',\lambda')$ be two principally polarised abelian varieties.
Let $\varphi : A \rightarrow A'$ be an isogeny.
We define the \textbf{dual isogeny of $\varphi$ with respect to the principal polarisations} as the isogeny $\tilde{\varphi} := \lambda^{-1} \circ \hat{\varphi} \circ \lambda' : A' \rightarrow A$.
We say that $\varphi : (A,\lambda) \rightarrow (A',\lambda')$ is an \textbf{$\bm{N}$-isogeny of principally polarised abelian varities} if $\tilde{\varphi}\circ \varphi = [N]$.
\end{definition}

Let $M$ be the matrix form of an isogeny between products of elliptic curves.
The \textbf{adjoint matrix} of $M$ is $\tilde{M} := (\hat{M}_{j,i})_{i,j \in \llbracket 1,n\rrbracket}$ which is the transpose of the matrix whose entries are the dual entries of $M$.
The dual isogeny, with respect to the product polarisations, of the isogeny given by $M$ has for matrix form the adjoint matrix of $M$.\\

The notions of \textbf{algorithms of evaluation of isogenies} and \textbf{representations of an isogeny} naturally extend to $N$-isogenies.
Notice that each $N$-isogeny is associated to some principal polarisations and thus algorithms of evaluation of $N$-isogenies also return the principal polarisation of the codomains.\\

Separable isogenies between elliptic curves are determined by their kernel (up to isomorphisms of the target curve). Given a kernel, the corresponding isogeny can be evaluated using Vélu's formulae, see \cite{velu_isogenies_1971}.
We have similar results for $N$-isogenies with $N$ prime to the characteristic of the field of definition.
Indeed, such isogenies are determined by their kernel, which is maximal isotropic, and there exists an analogue to Vélu's formulae for them.
This notion of maximal isotropy is central and requires to introduce the Weil pairing for principally polarised abelian varieties.

\begin{definition}[Polarised Weil pairing]
\label{def_weil}
Let $(A,\lambda)$ be a polarised abelian variety over a field and $N$ be prime to the characteristic of this field.
There exists a \textbf{canonical nondegenerate pairing} $e_N : A[N] \times \hat{A}[N] \rightarrow \mu_N(\bar{\mathbb{F}})$, where $\mu_N(\bar{\mathbb{F}})$ is the group of $N$th roots of $1$ in $\bar{\mathbb{F}}$.
This pairing is called the Weil $N$-pairing.
The \textbf{polarised Weil $N$-pairing} $e_{N,\lambda}$ is then the canonical nondegenerate pairing $A[N] \times A[N] \rightarrow \mu_N(\bar{\mathbb{F}}), (P,Q) \mapsto e_{N}(P,\lambda(Q))$.
\end{definition}

\begin{definition}[Maximal isotropic subgroup]
With the same notations as in Definition \ref{def_weil}.
Let $H$ be a proper subgroup of $A[N]$.
The subgroup $H$ is \textbf{maximal isotropic in $\bm{A[N]}$} if the polarised Weil pairing $e_{N,\lambda}$ is trivial over $H$ but is not over any proper supergroup of $H$.
For an isogeny of $A$ having a maximal isotropic kernel in $A[N]$ is equivalent to be an $N$-isogeny.
\end{definition}

\begin{lemma}[Proposition 1.1 in \cite{kani_number_1997}]
\label{lem_kanibij}
Let $(A,\lambda),(A',\lambda')$ and $(A'',\lambda'')$ be principally polarised abelian varieties such that there exist $\varphi' : (A,\lambda) \rightarrow (A',\lambda')$ and $\varphi'' : (A,\lambda) \rightarrow (A'',\lambda'')$ two $N$-isogenies with $\ker \varphi' = \ker \varphi''$, where $N$ is coprime to the characteristic of the abelian varieties' field of definition.  
Then there is an isomorphism $\gamma$ between $A'$ and $A''$ such that $\varphi'' = \gamma \circ \varphi'$ and $\lambda' = \widehat{\gamma} \circ \lambda'' \circ \gamma$, i.e. $\gamma : (A', \lambda') \rightarrow (A'',\lambda'')$ is a $1$-isogeny.
We say that $\gamma$ is an \textbf{isomorphism of principally polarised abelian varieties}.
\end{lemma}

In further results of this section, we shall need to recover endomorphisms of a given product of elliptic curves $E^n\times E'^n$ from its kernel.
Thus, it is important to have a description of the group of the automorphisms of $E^n\times E'^n$ as, by Lemma \ref{lem_kanibij}, endomorphisms with the same kernel differ only by an automorphism.

\begin{lemma}
\label{lem_aut}
Let $E_1$ and $E_2$ be two elliptic curves and $n,m$ be two integers.
Let $\Aut(E_1^n\times E_2^m,\lambda_{E_1^n\times E_2^m})$ be the group of automorphisms of the principally polarised abelian variety $(E_1^n\times E_2^m,\lambda_{E_1^n\times E_2^m})$.
Then for any element $\psi$ of $\Aut(E_1^n\times E_2^m,\lambda_{E_1^n \times E_2^m})$, we have
$$
	M(\psi) = \begin{pmatrix} A_1 & B_{1,2} \\ B_{2,1} & A_2 \end{pmatrix},
$$
where for any $i,j \in \{1,2\}$, $A_i$ is a matrix of dimension $n_i$ with entries in $\Aut(E_i)\cup\{0\}$ and $B_{i,j}$ is a matrix of dimension $n_i\times n_j$ with entries in $\Iso(E_j,E_i)\cup\{0\}$.
Moreover $M(\psi)$ contains only one non-zero entry per column and per row.
\end{lemma}

\begin{proof}
Let $\psi \in \Aut(E_1^{n_1}\times E_2^{n_2},\lambda_{E_1^{n_1} \times E_2^{n_2}})$.
As $\psi$ is an automorphism of a principally polarised abelian variety, we have $\hat{\psi} \circ \lambda \circ \psi = \lambda$ thus $\psi \tilde{\psi} = [1]$, which, in matrix form, gives $M(\psi) \tilde{M}(\psi) = \text{I}_{n_1+n_2}$.
Let us denote by $(\psi_{i,j})_{1\leq i,j \leq n_1+n_2}$ the coefficients of the matrix form $M(\psi)$.
For any $i \in \llbracket 1,n_1+n_2 \rrbracket$, we have
$$[1] = \sum_{j=1}^{n_1+n_2} \psi_{i,j} \circ \hat\psi_{i,j} = \sum_{j=1}^{n_1+n_2} [\text{degree}(\psi_{i,j})].$$
This implies that for any $i$, exactly one of the isogenies $\psi_{i,j}$ is non-zero, and that isogeny has degree one, hence it is an isomorphism.
The identity $\tilde{\psi} \psi = [1]$ yields the same results for columns.
Moreover, for $\psi$ to be well-defined, the domain and codomain of each $\psi_{i,j}$ must be 
\begin{align*}
\domain(\psi_{i,j}) = 
\begin{cases}
	E_1, &\text{if } 1 \leq j \leq n_1, \\
	E_2, &\text{if } n+1 \leq j \leq n_1+n_2,
\end{cases}\\
\codomain(\psi_{i,j}) =
\begin{cases}
	E_1, &\text{if } 1 \leq i \leq n_1, \\
	E_2, &\text{if } n+1 \leq i \leq n_1+n_2. \\
\end{cases}
\end{align*}
\end{proof}

As presented by Robert in \cite{robert_applications_2022}, isogenies between abelian varieties can be embedded into isogenies of higher dimensions.
Namely, given an isogeny $\varphi$ between abelian varieties, one can construct a higher dimensional isogeny such that one of its matrix form coefficients is equal to $f$, up to isomorphism.
This result is a generalisation of a construction in dimension 1 given by Kani in \cite{kani_number_1997}.

\begin{lemma}[Lemma 3.6 in \cite{EC:Robert23}]
\label{lem_ourkani}
Let $A$ and $B$ be two principally polarised abelian varieties of dimension $g$ over a base field of characteristic $p$.
Let $\varphi_1,\varphi_2'$ be two $d_1$-isogenies and $\varphi_2,\varphi_1'$ be two $d_2$-isogenies such that $(d_1+d_2,p)=1$ and $\varphi_1' \circ \varphi_1 = \varphi_2' \circ \varphi_2$ is a $d_1d_2$-isogeny from $A$ onto $B$, i.e.

\[
\begin{tikzcd}
	A \arrow{r}{\varphi_1} \arrow[swap]{d}{\varphi_2} & \varphi_1(A) \arrow{d}{\varphi_1'} \\
	\varphi_2(A) \arrow{r}{\varphi_2'} & B
\end{tikzcd}.
\]
Then 
$$
\begin{pmatrix} \varphi_1 & \widetilde{\varphi_1'} \\ -\varphi_2 & \widetilde{\varphi_2'} \end{pmatrix}
$$
is the matrix form of a $(d_1+d_2)$-isogeny $F: A \times B \rightarrow \varphi_1(A) \times \varphi_2(A)$.
Moreover, if $\gcd(d_1,d_2)=1$ then the kernel of $F$ is $\widetilde{F}(\varphi_1(A)[d_1+d_2]\times \{0\})$ and is of rank $2g$.
\end{lemma}

Lemma \ref{lem_embedding} below describes how an isogeny $\varphi$ of degree $N$ between elliptic curves can be embedded into an $N'$-endomorphism in dimension 8, for some $N' > N$.

This lemma exploits the recent techniques developed to attack SIDH~\cite{EC:CasDec23,EC:MMPPW23,EC:Robert23}.

\begin{lemma}
\label{lem_embedding}
Let $E_1$ and $E_2$ be two elliptic curves over a finite field $\mathbb{F}_{p^k}$ and $\varphi: E_1 \rightarrow E_2$ be an isogeny of degree $N$.
Let $N'>N$ be an integer such that $(N',Np) = 1$.
Let $m_1,m_2,m_3$ $m_4$ be integers such that $m_1^2 + m_2^2 + m_3^2 + m_4^2 = N'- N$ and let $\alpha_{E_1}$ (resp. $\alpha_{E_2}$) be the endomorphism over $E_1^4$ (resp. $E_2^4$) given by the matrix 
$$
\begin{pmatrix} m_1 & -m_2 & -m_3 & -m_4 \\ m_2 & m_1 & m_4 & - m_3 \\ m_3 & -m_4 & m_1 & m_2 \\ m_4 & m_3  & -m_2 & m_1  \end{pmatrix}.
$$
Let $H := \{(\tilde{\alpha}_{E_1}(P),\varphi^{\times 4}(P)) | P \in E_1^4[N']\}$; then there exists an $N'$-isogeny of $E_1^4\times E_2^4$ of kernel $H$.\\
Furthermore, the following holds for any $N'$-isogeny $G$ of $E_1^4\times E_2^4$ of kernel $H$.
\begin{itemize}
\item The codomain of $G$ is isomorphic to $(E_1^4\times E_2^4, \lambda_{E_1^4 \times E_2^4})$ as principally polarised abelian varieties.
\item For any isomorphism $\gamma : G(E_1^4 \times E_2^4,\lambda_{E_1^4 \times E_2^4}) \rightarrow (E_1^4 \times E_2^4,\lambda_{E_1^4\times E_2^4})$, there exist an integer $i \in \llbracket 1, 8 \rrbracket$ and an isomorphism $\psi$ in $\Iso(E_2,E')$, where $E' \in \{E_1,E_2\}$, such that the following diagram commutes
\[
\begin{tikzcd}
	E_1 \arrow{r}{G\circ\tau_1} \arrow[swap]{d}{\varphi} & G(E_1^4 \times E_2^4) \arrow{d}{\pi_i \circ \gamma} \\
	E_2 \arrow{r}{\psi} & E'
\end{tikzcd}
\]
i.e. $ \pi_i(\gamma(G(\tau_1(P)))) = \psi(\varphi(P)), \text{ for all } P \in E_1$, and thus $(\gamma \circ G,1,i)$ is an embedding representation of $\psi \circ \varphi$.\\
\end{itemize}
\end{lemma}

\begin{proof}
We use the same notations as above.\\
Since the matrix form of $\varphi^{\times 4}$ is diagonal, we have the following commutative diagram 
\[
\begin{tikzcd}
	E_1^4 \arrow{r}{\alpha_{E_1}} \arrow[swap]{d}{\varphi^{\times 4}} & E_1^4 \arrow{d}{\varphi^{\times 4}} \\
	E_2^4 \arrow{r}{\alpha_{E_2}} & E_2^4.
\end{tikzcd}
\]
By construction of $\varphi^{\times 4}$ is an $N$-isogeny and $\alpha_{E_1}$ and $\alpha_{E_2}$ are $(N'-N)$-isogenies.
Thus, the sum of the degree of $\varphi^{\times 4}$ with the degree of $\alpha_{E_1}$ or $\alpha_{E_2}$ is equal to $N'$.
By assumption, $N'$ is coprime to $Np$.
In particular $N'$ is coprime to $p$ and $N$ is coprime to $N'-N$.
Hence, by taking $A := E_1^4, B := E_2^4, \varphi_1:=\alpha_{E_1}, \varphi'_2 := \alpha_{E_2}, \varphi_2 := \varphi^{\times 4}$ and $\varphi_1':= \varphi^{\times 4}$, we have $d_1 = N'-N$, $d_2 = N$ and all the assumptions of Lemma \ref{lem_ourkani} are satisfied.
Its application gives us an $N'$-endomorphism $F$ of $E_1^4\times E_2^4$ with kernel 
$$\ker F = \{ (\tilde{\alpha}_{E_1}(P),\varphi^{\times 4}(P)) | P \in E_1^4[N']\}$$
and matrix form
$$
M(F) = \begin{pmatrix} M(\alpha_{E_1}) & M(\reallywidetilde{\varphi^{\times 4}}) \\ -M(\varphi^{\times 4}) & M(\reallywidetilde{\alpha_{E_2}}) \end{pmatrix}. 
$$ 
Then, for any $P \in E_1$, we have
\begin{align}
\label{align_HDembedding_1}
F(\tau_1(P)) = (m_1P,m_2P,m_3P,m_4P,-\varphi(P),0,0,0)
\end{align}
Let $G$ be an $N'$-isogeny of $(E_1^4\times E_2^4, \lambda_{E_1^4 \times E_2^4})$ with $\ker G = \ker F$.
By Lemma \ref{lem_kanibij}, $G(E_1^4 \times E_2^4,\lambda_{E_1^4 \times E_2^4})$ and $(E_1^4\times E_2^4,\lambda_{E_1^4\times E_2^4})$ are isomorphic and for any isomorphism $\gamma$ from $G(E_1^4 \times E_2^4,\lambda_{E_1^4 \times E_2^4})$ to $(E_1^4 \times E_2^4,\lambda_{E_1^4 \times E_2^4})$ there exists an automorphism $\psi$ of $E_1^4 \times E_2^4$ such that we have
\begin{align}
\label{align_HDembedding_2}
\gamma \circ G = \psi \circ F.
\end{align}

By Lemma \ref{lem_aut}, there exist $8$ isomorphisms $\psi_1,\dots\psi_8$ such that
$$
\psi_{i} \in \begin{cases} \Aut(E_1)\cup\Iso(E_1,E_2), & \text{ if } i \in \llbracket 1, 4 \rrbracket, \\ \Aut(E_2) \cup \Iso(E_2,E_1), & \text{ if } i \in \llbracket 5,8 \rrbracket \end{cases}
$$
and a map $\sigma$ permuting coordinates of the points of $E_1^4 \times E_2^4$ such that, for any point $(P_1,\dots,P_8)$ of $E_1^4\times E_2^4$, we have
\begin{align}
\label{align_HDembedding_3}
\psi(P_1,\dots,P_8) = \sigma(\psi_1(P_1),\dots,\psi_8(P_8)).
\end{align}

Let $i$ be the integer such that $\pi_i(\sigma(Q_1,\dots,Q_8)) = Q_5$ for any $(Q_1,\dots,Q_8) \in E_1^4 \times E_2^4$.
Then, for any $P \in E_1$,
\begin{align*}
	\pi_i(\gamma(G(\tau_1(P)))) &= \pi_i(\psi(F(\tau_1(P)))) \text{, by (\ref{align_HDembedding_2}),}\\
	&= \pi_i(\sigma( (\psi_1(m_1P),\dots,\psi_4(m_4P),\psi_5(-\varphi(P)),0,0,0))) \text{, by (\ref{align_HDembedding_1}),}\\
	&= -\psi_5(\varphi(P)) \text{, by construction $\pi_i$}.
\end{align*}
This conclude the proof as $-\psi_5$ is an element of $\Aut(E_2)\cup\Iso(E_2,E_1)$.
\end{proof}

This embedding can be evaluated efficiently using the analogue of Vélu's formulae in higher dimension introduced by Lubicz and Robert, \cite{lubicz_computing_2012}.

\begin{lemma}[\hspace{-0.45pt}\cite{robert_evaluating_2022}]
\label{lem_velu}
Let $(E_1\times \dots \times E_n, \lambda_{E_1 \times \dots \times E_n})$ be a principally polarised abelian varity where $E_i$ are elliptic curves defined over $\mathbb{F}_{p^k}$.
Let $H$ be a maximal isotropic subgroup of $(E_1 \times \dots \times E_n)[N']$ where $N'$ is an integer coprime to $p$ of prime factorisation $\prod_{i=1}^r \ell_i^{e_i}$.

Given $H$ as a set of generators living in $(E_1 \times \dots \times E_n)[\ell_i^{e_i}]$, one can compute a representation of an $N'$-isogeny $G$ of $(E_1\times \dots \times E_n, \lambda_{E_1 \times \dots \times E_n})$ with kernel $H$ such that
\begin{itemize}
	\item it takes $O(B^8D \log^2(N') \log(B))$ arithmetic operations over $\mathbb{F}_{p^k}$ to get this representation,
	\item the representation has size $O(k M \log(N') \log p)$ bits,
	\item the representation allows to evaluate $G$ on a point in $O(B^8 M \log(N') \log(B))$ operations over its field of definition, 
\end{itemize}
	where $B,M$ and $D$ are any bounds such that $B \geq P^*(N')$, $M \geq \max_{i=1}^r \delta_{E_i}(N')$ and $D \geq \max_{i=1}^r \delta_{E_i,2}(N')$.
\end{lemma}

\begin{proof}
This result is simply a rephrasing of \cite[4. The algorithm]{robert_evaluating_2022}.
The main idea behind achieving this complexity is to compute the higher dimensional isogeny as a chains of power prime isogenies.
In the original description of the algorithm, Robert uses \cite{lubicz_fast_2023} to compute representation of each of these isogenies.
For the same complexity, we suggest using \cite[Theorem 53]{EC:DLRW24} instead, as it provides a more convenient statement.
In particular, we have access to an explicit description of the theta coordinates of the output.
\end{proof}

\begin{remark}
In this section, we always assume that the isogenies are embedded into dimension 8.
Lemma \ref{lem_velu}, and so all the results derived from it, could be more efficient if the isogenies were embedded into dimensions 2 or 4, unfortunately, it is not always possible.
Indeed, for dimension 8, we decompose $N' - N$ as a sum of four squares to construct an endomorphism of $E^4$ using the Zarhin's trick, see \cite{zarhin_remark_1974}.
For dimension 2 (resp. 4), $N'-N$ needs to be a square (resp. a sum of two squares) to construct easily an endomorphism of $E$ (resp. $E^2$) and to embed the isogenies into dimension 2 (resp. 4).
It is possible to relax these conditions, under some heuristics and when the endomorphism ring is known.
Here, we neither want to rely on heuristics, nor presume that we know the endomorphism ring, so we only consider the case of dimension 8.
\end{remark}

By Lemma \ref{lem_embedding}, if an isogeny $\varphi: E_1 \rightarrow E_2$ is divisible by an integer $n$ then it can be embedded into an isogeny in dimension 8 of kernel 
$$
	H := \{(\tilde{\alpha}_{E_1}(P), (\varphi/n)^{\times 4}(P))| P \in E_1^4[N']\},
$$
with $N'$ and $\alpha_{E_1}$ defined as in the lemma.
Thanks to Lemma \ref{lem_velu}, one can compute in polynomial time this $8$-dimensional isogeny from its kernel.
It only remains to show how we can compute the kernel $H$ to get a complete division algorithm.

By construction $N'$ is an integer coprime to $\deg(\varphi)$ and $n$, then we have 
$$(\varphi/n)^{\times 4}(P) = (s\varphi)^{4 \times}(P) \text{ with } n^{-1} = s \mod N'.$$
This trick allows us to compute the kernel $H$ from an efficient representation of $\varphi$.
Lemma \ref{lem_maximalisotropic} ensures that computing $H$ this way always provides a maximal isotropic subgroup of $E_1^4 \times E_2^4$ even if $\varphi/n$ is not a well-defined isogeny.
This will be crucial to verify the divisibility of an isogeny by an integer.

\begin{lemma}
Let $\varphi: E_1 \rightarrow E_2$ be an isogeny.
Let $n^2$ be a divisor of $\deg(\varphi)$ and $N = \deg(\varphi)/n^2$.
Let $N' > N$ such that $(N',p\deg{(\varphi)}) = 1$ and $s = n^{-1} \mod N'$.
Let $\alpha_{E_1}$ be an $m$-endomorphism of $E_1^4$ with $m = N' - N$.\\
Then $H := \{ (\tilde{\alpha}_{E_1}(P), s\varphi^{\times 4}(P)) | P \in E_1^4[N'] \}$ is a maximal isotropic subgroup of $(E_1^4 \times E_2^4)[N']$.

\label{lem_maximalisotropic}
\end{lemma}

\begin{proof}
The subgroup structure of $H$ comes immediately by construction.
We claim that $H$ is maximal isotropic.
Let $\lambda_1$ be the product polarisation over $E^4_1$ and $\lambda_2$ be the product polarisation over $E_2^4$.
Let us show that the Weil pairing $e_{N',\lambda_1 \times \lambda_2}$ is trivial between $(\tilde{\alpha}_{E_1}(P), s\varphi^{\times 4}(P))$ and $(\tilde{\alpha}_{E_1}(Q),s\varphi^{\times 4}(Q))$ for any $P=(P_1,P_2,P_3,P_4)$ and $Q=(Q_1,Q_2,Q_3,Q_4)$ in $E_1^4[N']$.
We have 
\begin{align*}
	&e_{N',\lambda_1 \times \lambda_2}((\tilde{\alpha}_{E_1}(P),s\varphi^{\times 4}(P)),(\tilde{\alpha}_{E_1}(Q),s\varphi^{\times 4}(Q)))\\
	&= e_{N',\lambda_1}(\tilde{\alpha}_{E_1}(P),\tilde{\alpha}_{E_1}(Q)) \cdot e_{N',\lambda_2}(s\varphi^{\times 4}(P),s\varphi^{\times 4}(Q)) \\
	&= e_{N'}(\tilde{\alpha}_{E_1}(P), \lambda_1 \lambda_1^{-1} \tilde{\alpha}_{E_1} \lambda_1 (Q)) \cdot e_{N'}(P, \widetilde{s\varphi^{\times 4}} \circ \lambda_2 \circ s \varphi^{\times 4} (Q))\\
	&= e_{N'}(\alpha_{E_1} \tilde{\alpha}_{E_1}(P), \lambda_1(Q)) \cdot e_{N'}(P, \lambda_1 \circ \widetilde{s \varphi^{\times 4}} \circ s \varphi^{\times 4}(Q))\\
	&= e_{N',\lambda_1}([m](P),Q) \cdot e_{N'}(P,\lambda_1([s^2n^2N]Q))\\
	&= e_{N',\lambda_1}(P,Q)^m \cdot e_{N',\lambda_1}(P,Q)^{s^2n^2N}\\
	&= e_{N',\lambda_1}(P,Q)^{m + s^2n^2N} = 1 \text{, as } m + s^2n^2N \equiv 0 \mod N'. 
\end{align*}
Thus $H$ is isotropic with respect to the product polarisation.
Finally, it is also maximal since it has order $N'^{8}$ which is the square root of the order of $(E_1^4 \times E_2^4)[N']$, see \cite[p 233]{david_mumford_abelian_1970}.

\end{proof}

It is now possible to provide Algorithm \ref{alg_division} which efficiently divides isogenies by integers.
This algorithm is similar to those presented by Robert in \cite[4. The algorithm]{robert_evaluating_2022} and the section 4 of \cite{robert_applications_2022}.

\begin{algorithm}
\caption{\texttt{Isogeny division}}\label{alg_division}
\begin{flushleft}
\hspace*{\algorithmicindent} \textbf{ Input : }
An isogeny $\varphi: E_1 \rightarrow E_2$, where $E_1,E_2$ are elliptic curves defined over $\mathbb{F}_{p^k}$, and two integers $n$ and $N'> \deg(\varphi)$ such that $(N',p\deg(\varphi))= 1$\\ 
\hspace*{\algorithmicindent} \textbf{ Output : }
A representation of $\varphi/n$ if it is a well-defined isogeny, \texttt{False} otherwise.
\end{flushleft}
\begin{algorithmic}[1]

	\State Set $N \gets \deg(\varphi)/n^2$.
	\If {$N \not \in \mathbb{N}$} 
	\State \Return \texttt{False}
	\EndIf
	\State Set $m \gets N' - N$.
	\State Decompose $m$ as $m_1^2 + m_2^2 + m_3^2 + m_4^2$.\label{alg1_decompose}
	\State Set $M \gets \begin{pmatrix} m_1 & -m_2 & -m_3 & -m_4 \\ m_2 & m_1 & m_4 & -m_3 \\ m_3 & -m_4 & m_1 & m_2 \\ m_4 & m_3 & -m_2 & m_1 \end{pmatrix}.$
	\State Let $\alpha$ be the $m$-endomorphism over $E_1^4$ given by the matrix $M$.
	\State Let $\tilde{\alpha}$ be the dual isogeny of $\alpha$ with respect to the product polarisation.
	\State $s \gets n^{-1} \mod N'$. \label{alg1_mod}
	\State Compute a factorisation $\ell_1^{e_1}\dots\ell_r^{e_r}$ of $N'$. \label{alg1_fact}
	\State Compute bases $(P_{1,i},P_{2,i})$ of $E_1[\ell_i^{e_i}]$ for $i \in \llbracket 1,r\rrbracket$. \label{alg1_basis1}
	\State Set $(P_1,P_2) \gets (\sum_{i=1}^r P_{1,i}, \sum_{i=1}^r P_{2,i})$ a basis of $E_1[N']$. \label{alg1_basis2}
	\State Compute a representation of an $N'$-isogeny $F$ of $E_1^4 \times E_2^4$ of kernel
		$$\ker F = \{ (\tilde{\alpha}_{E_1}(\tau_i(P_j)), s \varphi^{\times 4} (\tau_i(P_j)))| \forall i \in \llbracket 1,4 \rrbracket, \forall j \in \{1,2\}\}.$$ \label{alg1_kernel}
\If {$F(E_1^4\times E_2^4, \lambda_{E_1^4 \times E_2^4}) \not \simeq (E_1^4 \times E_2^4, \lambda_{E_1^4 \times E_2^4})$}\label{alg1_check1}
		\State \Return \texttt{False}. 
	\EndIf
	\State Set $\gamma : F(E_1^4\times E_2^4, \lambda_{E_1^4\times E_2^4}) \rightarrow (E_1^4 \times E_2^4, \lambda_{E_1^4 \times E_2^4})$ be an isomorphism of principally polarised abelian varieties. \label{alg1_iso}
	\If {$E_1 \simeq E_2$} \label{alg1_iso_1D}
		\State Set $\psi_0:E_2 \rightarrow E_1$ to be an isomorphism.
	\Else 
		\State Set $\psi_0: E_2 \rightarrow E_1$ to be the zero map.
	\EndIf
	\State Compute the sets of maps $S_{E_1} := \Aut(E_1)\psi_0$ and $S_{E_2} := \Aut(E_2)$. \label{alg1_aut}
	\For {$t \in \llbracket 1, 8 \rrbracket$} \label{alg1_for1}
		\For {$\psi \in S_{E_1}$, if $1 \leq t \leq 8$, or $\psi \in S_{E_2}$, if $5 \leq t \leq 8$,}
			\If {$n(\psi^{-1} \circ \pi_t \circ \gamma \circ F \circ \tau_1(P_{i,j})) = \varphi(P_{i,j}),\forall i \in \{1,2\}, \forall j \in \llbracket 1,r \rrbracket$} \label{alg1_if1}
				\State \Return The representation of $\varphi/n$ induced by $\psi^{-1} \circ \pi_t \circ \gamma \circ F \circ \tau_1$.
			\EndIf
		\EndFor
	\EndFor	
	\State \Return \texttt{False} \label{alg1_end}
\end{algorithmic}
\end{algorithm}

\begin{theorem}
\label{the_division}
Algorithm \ref{alg_division} is correct and runs in 
\begin{itemize}
	\item $O(\max(M^2,D) B^8 \log^2( N' )\log(B))$ operations over $\mathbb{F}_{p^k}$,
	\item plus the cost of the factorisation of $N'$,
	\item plus the cost of the computation of the bases of $E_1[\ell^e]$ for each prime power divisor $\ell^e$ of $N'$,
	\item plus the cost of $O(\log N')$ evaluations of $\varphi$ over these bases,
	\item plus the cost of decomposing $N' - N$ as a sum of four squares (which takes $O(\log^2 N')$ arithmetic operations over integers),
\end{itemize}
where $B,M,D$ give the following bounds $B \geq P^*(N'), M \geq \delta_E(N')$ and $D \geq \delta_{E,2}(N')$.
Moreover, if $\varphi/n$ is indeed an isogeny, the output representation of $\varphi/n$ has the following properties:
\begin{itemize}
	\item It has size $O(k M \log(N')\log(p))$ bits.
	\item It allows to evaluate $\varphi/n$ in $O(B^8 M \log(N') \log (B))$ operations over the field of definition of the input.
\end{itemize}
\end{theorem}

\begin{proof}
Let us prove the correctness of Algorithm \ref{alg_division}.

First, by Lemma \ref{lem_maximalisotropic}, $\ker F$ is always a maximal isotropic subgroup of $(E_1^4 \times E_2^4)[N']$ and thus the isogeny $F$ is well defined.

When $\varphi/n : E_1 \rightarrow E_2$ is an isogeny, we have that $(\varphi/n)|_{E[N']} = (s\varphi)|_{E[N']}$.
Hence, by Lemma \ref{lem_embedding}, $F$ is isomorphic to an $N'$-isogeny that embeds $\varphi/n$.
More precisely, $F(E_1^4\times E_2^4,\lambda_{E_1^4 \times E_2^4}) \simeq (E_1^4 \times E_2^4,\lambda_{E_1^4 \times E_2^4})$ as principally polarised abelian varieties and for any isomorphism $\psi$ between them, there exist an isomorphism $\psi: E_2 \rightarrow E'$, where $E' \in \{ E_1, E_2\}$, and an integer $t \in \llbracket 1,8 \rrbracket$ such that
\begin{align}
\label{align_search}
	\pi_t(\gamma(F(\tau_1(P)))) = \psi(\varphi/n)(P), \forall P \in E_1.
\end{align}

We check at line \ref{alg1_check1} if we can find such isomorphism $\gamma$.
If this is the case, we fix an isomorphism $\gamma:F(E_1^4 \times E_2^4, \lambda_{E_1^4 \times E_2^4}) \rightarrow (E_1^4 \times E_2^4, \lambda_{E_1^4 \times E_2^4})$.
Otherwise, by Lemma \ref{lem_embedding}, $\varphi/n$ is not an isogeny and the algorithm must return \texttt{False}.

We then look for an isomorphism $\psi$ and an integer $t$ verifying Equality (\ref{align_search}).
To satisfy the equality, $\psi$ needs to have the same codomain as $\pi_t \circ \gamma \circ F$ which is equal to $E_1$ if $1 \leq t \leq 4$ and $E_2$ if $5 \leq t \leq 8$.
In the first case, $\psi$ is an element of $\Aut(E_1)\psi_0$, where $\psi_0: E_2 \rightarrow E_1$ is an isomorphism.
In the second case, $\psi$ is an automorphism of $E_2$.

In the for loop, we search for a solution $(\psi,t)$ of Equality (\ref{align_search}) over the bases of $E_1[\ell_i^{e_i}], \forall i \in \llbracket 1,r\rrbracket$.
It is equivalent to checking Equality (\ref{align_search}) over $E_1[N']$ as $(\ell_i,\ell_j) = 1, \forall i \not = j \in \llbracket 1,r\rrbracket$.
Moreover as $N' > \deg \varphi$ and $\deg \varphi \geq 2$, a solution of (\ref{align_search}) over $E_1[N']$ is a solution over the entire elliptic curve $E_1$.
Indeed, if two isogenies $\varphi$ and $\varphi'$ of same degree are equal over $E_1[N']$ with $N' > \deg \varphi \geq 2$, then they are equal everywhere.
Let us assume that $\phi := \varphi - \varphi'$ is a non-zero isogeny.
We have $\phi(E_1[N']) = 0$, hence
$$
4\deg\varphi \leq \deg \varphi^2 < {N'}^2 = \#E[N'] \leq \deg \phi \leq ( (\deg \varphi)^{1/2} + (\deg\varphi')^{1/2})^2 = 4\deg\varphi,
$$
this is a contradiction.
Notice that we assumed that $\deg \varphi \geq 2$.
In fact, we can even assume that $\deg \varphi \geq 4$, otherwise $\deg(\varphi)/n^2$ is not an integer when $n>1$.
Moreover, the division is trivial if $n$ =1.

Since we are doing an exhaustive search at line \ref{alg1_for1}, if $\varphi/n$ is an isogeny, the algorithm will find an embedding representation $(F,1,t)$ of $\varphi/n$ up the two isomorphims $\psi$ and $\gamma$.
If no such coefficient of $F$ is found, Lemma \ref{lem_embedding} implies that $\varphi/n$ is not an isogeny.\\
The output representation of $\varphi/n$ is then given by the composition of the representation of $\psi^{-1}$ with the embedding representation $(\gamma \circ F,1,t)$.\\

Let us now turn to the complexity analysis of the different steps.
We consider the following bounds $B \geq P^*(N'), M \geq \delta_E(N'), D \geq \delta_{E,2}(N')$.

\begin{itemize}

	\item[{}]
	{[1-\ref{alg1_mod}]:}\\
	The decomposition of $m$ at the line \ref{alg1_decompose} can be done in $O(\log^2 N')$ arithmetic operations over the integers, see \cite{pollack_finding_2018}.
	This is the only complexity of the algorithm where operations are not counted over the finite field but over integers.
	The computational cost of the other lines is negligible compared to the rest of the algorithm.

	\item[{}]
	{[\ref{alg1_fact} - \ref{alg1_basis1}]:}\\
	We do not estimate the complexity of these steps now, we simply acknowledge them in the overall analysis.

	\item[{}]
	{[\ref{alg1_basis2}-\ref{alg1_iso}]:}\\
	At line \ref{alg1_basis2}, we denote a basis of $E_1[N']$ by $(P_1,P_2)$ only formally to get simple notations.
	The computation are always done with the $(P_{1,i},P_{2,i})$, where $i \in \llbracket 1,r\rrbracket$.

	By Lemma \ref{lem_velu}, getting a representation of the $N'$-isogeny $F$ takes :
	\begin{itemize}
		\item $O(B^8 D \log^2(N')\log(B))$ arithmetic operations over $\mathbb{F}_{p^k}$,
		\item $O(\log N')$ evaluations of $\varphi$ over the bases of $E[\ell_i^{e_i}]$, for $i \in \llbracket 1,r\rrbracket$.
	\end{itemize}

	Then, it remains to compute an isomorphism $\gamma$ between the codomain of $F$ and $(E_1^4 \times E_2^4, \lambda_{E_1^4 \times E_2^4})$.
	As described in \cite[Appendix F.3]{EC:DLRW24}, one can perform an exhaustive search over the symplectic group $\Sp_{16}(\mathbb{Z}/4\mathbb{Z})$, which consists of the set of $16 \times 16$ matrices that preserve symplectic form, to find a matrix that sends the theta coordinates of the codomain of $F$ to the theta coordinates of $(E_1^4 \times E_2^4, \lambda_{E_1^4 \times E_2^4})$ (see Remark \ref{rem_theta} for the construction of the latter).
	This step does not impact the overall complexity, as the number of possibilities is constant for a given dimension.\footnote{
	As this constant is very large, it is preferable, for more practical usage, to directly compute the matrix we are looking for, see \cite[Appendix F.3]{EC:DLRW24}.}
	If no such linear transformation is found, then the two principally polarised abelian varieties are not isomorphic.

	\item[{}]
	{[\ref{alg1_iso_1D} - \ref{alg1_aut}]:}\\
	Those steps are done efficiently using basics of elliptic curves' theory, see \cite{silverman_arithmetic_1986}.

	Checking if two elliptic curves are isomorphic can be done using the $j$-invariant.
	Then one can get an explicit isomorphism with a small computation from the Weierstrass equations of the elliptic curves.

	The group of automorphisms of an elliptic curve $E$, for $p > 3$, is generated by
	\begin{equation*}
	\begin{cases}
		\{ (x,y) \mapsto (x,-y), (x,y) \mapsto (-x,iy) \} & \text{if $j(E) = 0$,}\\
		 \{ (x,y) \mapsto (x,-y), (x,y) \mapsto (\zeta_3 x,y) \} & \text{if $j(E) = 1728$,}\\
		 \{ (x,y) \mapsto (x,-y)\} & \text{otherwise,}
    	\end{cases}       
	\end{equation*}
	where $i$ is a primitive $2$-nd root of unity and $\zeta_3$ is a primitive $3$-rd root of unity, both in $\bar{\mathbb{F}}_{p}$.
	When $p = 2$ or $3$, there are at most respectively $24$ and $12$ automorphisms.
	One can find their explicit description in \cite[Appendix A]{silverman_arithmetic_1986}.

	\item[{}]
	{[\ref{alg1_for1} - \ref{alg1_end}]:}\\
	The loop at line \ref{alg1_for1} has $O(\log N')$ iterations where the evaluations of $\tau_1,\gamma,\pi_t,\psi^{-1}$ are negligeable. 
	Thus it takes in total $O(\log N')$ evaluations of $\varphi$ over $E_1[\ell_i^{e_i}]$, $\forall i \in \llbracket 1, r \rrbracket$, plus $O(B^8 M \log^2(N') \log(B))$ operations over extension of degree at most $M$ to evaluate $F$.
\end{itemize}

We get the claimed complexity by summing all those steps.

In addition, the size of the ouput representation of $\varphi/n$ is mainly the size of the representation of $F$ thus it has size $O(k M \log(N') \log p)$ bits.

Finally, it allows to evaluate $\varphi/n$ at a point in $O(B^8 M \log(N') \log(B))$ operations over its field of definition because all the computations are negligeable in comparison to the evaluation of $F$.
\end{proof}

When the input of Algorithm \ref{alg_division} is efficiently represented, it leads to Theorem \ref{cor_division} which concludes this section about efficient division of isogenies.

\begin{proof}[Proof of Theorem \ref{cor_division}]
To get this result, one only needs to find a suitable powersmooth integer $N'$ and to take advantage of the efficient representation of $\varphi$.
One can use the approach proposed in \cite{robert_applications_2022} to get such an integer $N'$:
we compute $N'$ by multiplying successive primes, coprime to $p\deg(\varphi)$, until their product is greater than $\deg(\varphi)$.
It takes $O(\log \deg(\varphi))$ arithmetic operations and gives an integer $N'$ which is $O(\log \deg(\varphi))$-powersmooth and such that $\log N' = O(\log \deg(\varphi))$.
Then, with the same notation as in Theorem \ref{the_division}, we have $M = B^2$ and $D = B^4$ which directly gives the claimed size of the representation of $\varphi/n$ and also the complexity to evaluate it.

Finally, by construction of $N'$ and because $\varphi$ is efficiently represented, all the remaining costs of Algorithm \ref{alg_division} are polynomial in $\log p$, $\log \deg(\varphi)$.
\end{proof}
\end{section}

\begin{section}{Solving \Primitivisation}

The \Primitivisation\ problem has been introduced recently in \cite{arpin_orienteering_2023} together with a quantum subexponential algorithm solving it.
However, it can be seen as a generalisation of the important problem of computing the endomorphism ring of an ordinary elliptic curve. 
Indeed, for ordinary elliptic curves, the Frobenius endomorphism $\pi$ is non-scalar, hence we always have an orientation by $\mathbb Z[\pi]$, and the endomorphism ring is a quadratic order containing $\mathbb Z[\pi]$.
Therefore computing the endomorphism ring of an ordinary curve really is a case of the \Primitivisation\ problem.

One initial idea to solve \Primitivisation\ is to adapt the best algorithms solving the ordinary version of \EndRing.
Until recently, these algorithms had a subexponential complexity.
Techniques involving higher dimensional isogenies changed the state of the art:
in \cite[Section~4]{robert_applications_2022}, Robert shows how to compute the endomorphism ring of ordinary elliptic curves in polynomial time (when the factorisation of the discriminant of the Frobenius is known), essentially by \emph{dividing} translates of the Frobenius.

In this section, we describe how one can solve the \Primitivisation\ Problem by adapting Robert's method and we give, as a direct consequence of this result, a polynomial algorithm to compute action of smooth ideals.

First, Theorem \ref{the_primitivisation} and its proof describe the algorithm and its complexity without assuming anything on the representation of the input endomorphism.
Notice that it requires computations only over a large enough torsion subgroup.
Hence, the complexity depends on the degree of the extension where this torsion lives and on the difficulty to evaluate the endomorphism on it.
Then, Corollary \ref{cor_primitivisation} specifies this theorem to the case where the endomorphism is given in efficient representation.
The two results assume that the factorisation of the discriminant of the order generated by the endomorphism is known.

\label{sec_prim}
\begin{theorem}[Primitivisation]
\label{the_primitivisation}
There exists an algorithm that takes as input:
\begin{itemize}
	\item A supersingular elliptic curve $E$ defined over a finite field $\mathbb{F}_{p^2}$,
	\item An endomorphism $\theta \in \End(E)\setminus \mathbb{Z}$ of degree $N$ together with the factorisation of $\disc(\mathbb{Z}[\theta])$,
	\item An integer $N'>N$ such that $(N',pN) = 1$ with three bounds $B \geq P^*(N')$, $M \geq \delta_E(N')$ and $D \geq \delta_{E,2}(N')$,
\end{itemize}
and returns a primitive orientation $\iota: \mathfrak{O} \hookrightarrow \End(E)$ such that $\mathbb{Z}[\theta] \subseteq \mathfrak{O}$.
In particular, the order $\mathfrak{O}$ is generated by an element $\omega$ such that
\begin{itemize}
	\item The orientation $\iota$ takes $O(M \log(N') \log(p))$ bits to store,
	\item The endomorphism $\iota(\omega)$ can be evaluated at a point in $O(B^8 M \log(N') \log(B))$ operations over its field of definition. 
\end{itemize}
This algorithm runs in $O(\max(M^2,D)B^8 \log^2(N')\log(N)\log(B))$ operations over $\mathbb{F}_{p^2}$, plus the cost of the computation of the bases $E[\ell^e]$ for each prime power divisor $\ell^e$ of $N'$ plus the cost of the computation of $O(\log N')$ evaluations of $\theta$ over these bases plus the cost of decomposing $N'-N$ as a sum of four squares.
\end{theorem}

\begin{proof}
Let $\alpha \in \bar{\mathbb{Q}}$ be a root of the minimal polynomial of $\theta$ and $\iota : \mathbb{Z}[\alpha] \hookrightarrow \End(E)$ be the orientation defined by $\iota(\alpha) = \theta$.
Let $K = \mathbb{Q}(\alpha)$, $f_\alpha$ be the conductor of the order $\mathbb{Z}[\alpha]$ and $O_K$ be the ring of integers $K$.
The factorisation of the conductor $f_\alpha$ can be deduced from the known factorisation of $\disc(\mathbb{Z}[\theta])$.
Indeed, let $\Delta_K$ be the discriminant of $K$ which is given by
\begin{equation*}
    \Delta_K=
    \begin{cases}
	    d, & \text{ if } d \equiv 1 \mod 4 \\
	    4 d, & \text{otherwise},
    \end{cases}
\end{equation*}
where $d$ is the squarefree part of $\disc(\mathbb{Z}[\alpha])$.
The integer $d$ is easy to compute since we have the factorisation of $\disc(\mathbb{Z}[\theta]) = \disc(\mathbb{Z}[\alpha])$.
As $f_\alpha^2 = \disc(\mathbb{Z}[\alpha])/\Delta_K$ one can directly deduce the factorisation of $f_\alpha$.\\

Let $\mathfrak{O} \subseteq O_K$ be the largest order such that $\iota$ extends to an embedding $\mathfrak{O} \hookrightarrow \End(E)$.
That embedding is the primitivisation of $\iota$, so the algorithm aims at determining $\mathfrak{O}$.
The inclusions $\mathbb{Z}[\alpha] \subseteq \mathfrak{O} \subseteq O_K$ suggest that $\mathfrak{O}$ can be determined by starting from $\mathbb{Z}[\alpha]$, and testing if the orientation can be extended locally at each prime factor of the conductor, as in the computation of the endomorphism ring of ordinary elliptic curves (see~\cite{robert_applications_2022}). This is described in Algorithm~\ref{algPrimitivisation}.

\begin{algorithm}
\caption{\Primitivisation}\label{algPrimitivisation}
\begin{flushleft}
\hspace*{\algorithmicindent} \textbf{ Input :}
$E$ a supersingular elliptic curve, $\iota : \mathbb{Z}[\alpha] \hookrightarrow \End(E)$ an orientation such that $\iota(\alpha) = \theta$ is a non scalar endomorphism of degree $N$, an integer $N'$ such that $N'>N$ and $(N',pN)=1$ and the factorisation of $f_\alpha$ the conductor of $\mathbb{Z}[\alpha]$.\\
\hspace*{\algorithmicindent} \textbf{ Output :}
A pair $(\alpha',\theta')$ describing the primitivisation of $\iota$. 
\end{flushleft}
\begin{algorithmic}[1]

	\State $t \gets  \bar{\alpha} + \alpha$.
	\State $\alpha' \gets 2\alpha - t$. \Comment{$\mathbb{Z}[\alpha'] = \mathbb{Z}[2\alpha]$.}
	\State $\theta' \gets 2\theta - [t]$.
	\State $(\ell_i)_{i = 1}^{n} \gets $ the list of distinct prime factors of $2f_\alpha$.
	\For {$i \in \llbracket 1,n \rrbracket$}
		\While {$\theta'/\ell_i \in \End(E)$}\Comment{using Algorithm \ref{alg_division} with the input $(\theta',\ell_i,N')$.}
	    	\State $\alpha' \gets \alpha'/\ell_i$.
	    	\State $\theta' \gets \theta'/\ell_i$.
        	\EndWhile
	\EndFor
	
	\If {$(\theta' + 1)/2 \in \End(E)$}\Comment{using Algorithm \ref{alg_division} with the input $(\theta'+1,2,N')$.}
		\State $(\alpha',\theta') \gets ((\alpha' + 1)/2,(\theta' + 1)/2)$.
	\EndIf
	\State \Return $(\alpha',\theta')$.
\end{algorithmic}
\end{algorithm}

Let us prove that Algorithm \ref{algPrimitivisation} is correct.
Write $t = \alpha + \bar{\alpha}$.
Since $\disc(\mathbb{Z}[\alpha]) = t^2 - 4\alpha\overline{\alpha}$, we have
$$\alpha' := 2\alpha - t = \pm \sqrt{\disc(\mathbb{Z}[\alpha])} = \pm f_\alpha \sqrt{\Delta_K} \text{, where $\Delta_K$ is the discriminant of } K.$$
Also define $\theta' := 2\theta -[t]$.
Note that for any divisor $m\mid f_\alpha$, we have $\mathbb{Z}[(f_\alpha/m)\sqrt{\Delta_K}] \subseteq \mathfrak{O}$ if and only if $\alpha'/[m] \in \End(E)$.
The for-loop of Algorithm \ref{algPrimitivisation} finds the largest such integer $m$, hence the resulting pair $(\alpha'/m,\theta'/m)$ satisfies $\mathbb{Z}[\alpha'/m] = \mathfrak{O} \cap \mathbb{Z}[\sqrt{\Delta_K}]$.

The case $\ell_i = 2$ and the final if-statement account for the fact that $\mathbb{Z}[\sqrt{\Delta_K}]$ is not the maximal order: it has conductor $2$ (we have $O_K = \mathbb{Z}[\sqrt{\Delta_K}/2]$ if $\Delta_K \equiv 0 \bmod 4$ and $O_K = \mathbb{Z}[(\sqrt{\Delta_K} + 1)/2]$ if $\Delta_K \equiv 1 \bmod 4$). That final correction accounted for, we actually obtain $\mathbb{Z}[\alpha'] = \mathfrak{O}$.\\

Let us now describe the complexity of Algorithm \ref{algPrimitivisation}.
			
Since $f_\alpha \leq 4N(\alpha) = 4N$, there are $O(\log(N))$ divisions using Algorithm \ref{alg_division}.
By Theorem \ref{the_division}, both checking if $\theta'/p_i \in \End(E)$ and getting a representation of the new endomorphism can be done in $O(\max(D,M^2) B^8 \log^2 (N' ) \log(B))$ operations over $\mathbb{F}_{p^2}$ plus the cost of the computation of the bases $E[\ell^e]$ for each prime power divisor $\ell^e$ of $N'$ plus the cost of the computation of $O(\log N')$ evaluations of $\theta'$ over these bases plus the cost of decomposing $N'-N$ as a sum of four squares.

After each update of the generating endomorphism using this theorem, the length of the representation will be in $O(M \log(N') \log(p))$ bits and will allow to evaluate a point in $O(B^8 M \log(N')\log(B))$ operations over the field of definition of the input.
Hence, all the divisions after the first one will run in $O(\max(M^2,D)B^8 \log^2(N')\log(B))$  operations over $\mathbb{F}_{p^2}$ and the final output will have the claimed properties.\\

It leads to a global complexity in $O(\max(M^2,D)B^8 \log^2(N')\log(N)\log(B))$ operations over $\mathbb{F}_{p^2}$ plus the cost of the computation of the torsion group bases and the $O(\log N')$ evaluations of $\theta$ over them plus the cost of the decomposing of $N'-N$ as a sum of four squares.

\end{proof}

Corollary \ref{cor_primitivisation} demonstrates that \Primitivisation\ can be solved in polynomial time when the input is efficiently represented by applying Theorem \ref{the_primitivisation}.

\begin{corollary}
\label{cor_primitivisation}
There exists an algorithm that takes as input:
\begin{itemize}
	\item A supersingular elliptic curve $E$ defined over a finite field $\mathbb{F}_{p^2}$,
	\item An endomorphism $\theta \in \End(E)\setminus \mathbb{Z}$ of degree $N$ together with the factorisation of $\disc(\mathbb{Z}[\theta])$,
\end{itemize}
and returns an efficiently represented primitive orientation $\iota: \mathfrak{O} \hookrightarrow \End(E)$ with $\mathfrak{O} \supseteq \mathbb{Z}[\theta]$ such that the orientation $\iota$ takes $O(\log^3(N) \log(p))$ bits to store.
This algorithm runs in time polynomial in $\log N$ and $\log p$.
\end{corollary}

\begin{proof}
As for Theorem \ref{cor_division}, this corollary is obtained by computing a suitable powersmooth $N'$ and by taking $M = B^2$ and $D = B^4$.
One also needs to use the fact that we are dealing with efficiently represented isogenies.
\end{proof}

\begin{subsection}{A direct application: Computing the action of smooth ideals}
\label{sec_smoothactions}
The class group action over $SS_\mathfrak{O}(p)$ is a key notion of the orientations' theory and is a central ingredient of the algorithms presented in Section \ref{sec_vect} to solve the oriented problems we are interested in, such as the $\mathfrak{O}$-\Vectorisation\ problem.

In the previous state of the art, the classical method to evaluate the action of an ideal $\mathfrak{a}$ on an $\mathfrak{O}$-oriented elliptic curve $(E,\iota)$ was to factor $\mathfrak{a}$ into primes ideals and then to apply the action of each factor ideals successively before dividing by the degree.
This method takes a time polynomial in the largest prime power factor of the norm of $\mathfrak{a}$ and in length of the input.
The quality of the representation obtained "degrades" as actions are computed (so applying the action of several powersmooth ideal iteratively would actually take exponential time).

We emphasise that the powersmooth constraint and the progressive "degradation" of the representations come from the division by the norm of $\mathfrak{a}$ in the computation of the induced orientation
$$
{\varphi_{\mathfrak{a}}}_*(\iota)(\omega) = (\varphi_\mathfrak{a} \circ \iota(\omega) \circ \hat \varphi_\mathfrak{a}) \otimes \frac{1}{N(\mathfrak{a})}, \text{where $\omega$ is a generator of } \mathfrak{O}.
$$
Indeed, an efficient representation of the induced orientation times the norm of $\mathfrak{a}$ is given by the composition of the efficient representations of $\varphi_\mathfrak{a},\iota(\omega)$ and $\hat\varphi_\mathfrak{a}$ as
$$N(\mathfrak{a}) {\varphi_{\mathfrak{a}}}_*(\iota) = \varphi_\mathfrak{a} \circ \iota(\omega) \circ \hat{\varphi}_\mathfrak{a},$$
where efficient representations of $\varphi_\mathfrak{a}$ and $\hat \varphi_\mathfrak{a}$ can be obtained in time polynomial in a smooth bound of the norm of $\mathfrak{a}$. 
This representation now "degrades" linearly with the number of successive actions and there are no more powersmooth constraints.
The issue with this division-free computation is that the induced orientation by the order $\mathbb{Z} + N(\mathfrak{a})\mathfrak{O}$ is obviously not primitive anymore, as it can still be divided by the norm of $\mathfrak{a}$.
Thanks to Corollary \ref{cor_primitivisation}, primitivising this orientation can be done efficiently given the factorisation of the conductor of the order, i.e. the factorisation of $N(\mathfrak{a})$ here.
This provides a new polynomial algorithm to compute action of smooth ideals.

\begin{corollary}
\label{cor_smoothactions}
Let $(E,\iota)$ be an $\mathfrak{O}$-oriented supersingular elliptic curve defined over $\mathbb{F}_{p^2}$.
Let $\mathfrak{a}$ be an invertible $\mathfrak{O}$-ideal of $B$-smooth norm.
If $(E,\iota)$ is efficiently represented, then one can compute an efficient representation of $[\mathfrak{a}] \star (E,\iota)$ in time polynomial in $B$, $\log p$, $\log N(\mathfrak{a})$ and in the length of the representation of $\iota$.
\end{corollary}

\begin{proof}
The prime factorisation $\ell_1^{e_1} \dots \ell_m^{e_m}$ of the norm of $\mathfrak{a}$ can be computed in time polynomial in $B$.
From this factorisation, one deduces the decomposition of $\mathfrak{a}$ as a product of $e_1$ prime ideals of norm $\ell_1$ with $e_2$ prime ideals of norm $\ell_2$ and so on.
Then, using a process similar to the \texttt{CSIDH} evaluation algorithm, generalised to a arbitrary orientations as in \cite{colo_orienting_2020}, one computes an efficient representation of $\varphi_\mathfrak{a}$ and $\hat \varphi_\mathfrak{a}$ in time polynomial in $B$, $k \log p$ and in the length of the representation of $\iota$.
Finally, one just has to primitivise the orientation $\varphi_\mathfrak{a} \circ \iota(\omega) \circ \hat \varphi_\mathfrak{a}$ using Corollary \ref{cor_primitivisation} to get an efficient representation of $\mathfrak{a} \star (E,\iota)$ in polynomial time.
\end{proof}

\begin{remark}
We remind the reader that the authors' objective in this article is to develop rigorous algorithms under GRH.
In this section, our aim is to provide an algorithm solving \Primitivisation\ and to give applications of this resolution, such as computing smooth ideal actions with better complexity than the current state of the art without relying on any heuristics.
Nevertheless, one might also be interested in using higher dimensional isogenies as a practical approach to computing group actions induced by oriented elliptic curves, even at the cost of relying on heuristics.
This direction is studied, for example, in SCALLOP-HD \cite{PKC:CheLerPan24}.
\end{remark}

To compute the action of any ideal, we shall use the polynomial time \texttt{Clapoti} algorithm introduced in \cite{page_introducing_2023} to avoid the smoothness constraint.

\begin{proposition}[{\cite[Theorem 2.9]{page_introducing_2023}}]
Let $(E,\iota)$ be an $\mathfrak{O}$-oriented supersingular elliptic curve defined over $\mathbb{F}_{p^2}$.
Given an integral invertible $\mathfrak{O}$-ideal $\mathfrak{a}$, one can compute an efficient representation of $\varphi_\mathfrak{a}: E \rightarrow E_\mathfrak{a}$ in probabilistic time polynomial in $\log p$, $\log N(\mathfrak{a})$ and in the length of the representation of $\iota$.
\end{proposition}

\begin{corollary}
\label{cor_clapoti}
Let $(E,\iota)$ be an $\mathfrak{O}$-oriented supersingular elliptic curve defined over $\mathbb{F}_{p^2}$.
One can compute an efficient representation of $[\mathfrak{a}] \star (E,\iota)$ in probabilistic time polynomial in $\log p$, $\log N(\mathfrak{a})$ and in the length of the representation of $\iota$.
\end{corollary}

\begin{remark}
The first version of the present article, which appeared prior to \texttt{Clapoti}, presented an algorithm for computing the action of any ideal with a rigorously proven complexity under GRH.
This complexity matched the previous state of the art while removing all heuristic assumption that were previously required.
The approach involved computing a smooth representative of the input ideal using algorithms with a subexponential complexity proven under GRH, such as \cite[Algorithm 1]{childs_constructing_2014}, followed by the application of Corollary \ref{cor_smoothactions}.
\end{remark}

\end{subsection}
\end{section}

\begin{section}{Resolution of $\mathfrak{O}$-\Vectorisation\ and $\alpha$-\EndRing}
\label{sec_vect}
In this section, we prove, under GRH only, the complexity of a classical and a quantum resolution of $\mathfrak{O}$-\Vectorisation\ which are as good as the current best algorithms based on heuristics.
We then use these rigorous solutions to solve the $\alpha$-\EndRing\ problem.

\begin{subsection}{Classical algorithm}
Currently, the best complexity we can expect for a classical algorithm solving the $\mathfrak{O}$-\Vectorisation\ problem is $l^{O(1)}|\disc(\mathfrak{O})|^{1/4}$, with $l$ the length of the input, for instance with a meet-in-middle approach as in \cite{DCC:DelGal16}.
Before the results presented in this paper, such complexity analyses were based on heuristics.
Indeed, one needs to compute multiple actions of ideals to solve $\mathfrak{O}$-\Vectorisation\ and without using higher dimensional isogenies to compute efficiently smooth ideal actions, one could only handle powersmooth ideals.
Thus, one had to assume some heuristics about the distribution of powersmooth ideals. 
Thanks to Corollary $\ref{cor_smoothactions}$, it is now possible to get rid of the constraint on powersmoothness and to rigorously prove this complexity.\\ 

To solve $\mathfrak{O}$-\Vectorisation, we first study the \Effective{\mathfrak{O}} problem where one also asks the $\mathfrak{O}$-ideal to send the orientation of the first $\mathfrak{O}$-oriented elliptic curve to the orientation of the second one.
Moreover, we want to be able to evaluate the isogeny induced by this ideal on another given $\mathfrak{O}$-orientable elliptic curve.
Notice that $\mathfrak{O}$-\Vectorisation\ and \Effective{\mathfrak{O}} are in fact both equivalent to $\mathfrak{O}$-\EndRing, see \cite{EC:Wesolowski22}.

\begin{problem}[\Effective{\mathfrak{O}}]
Given $(E,\iota),(E',\iota'),(F,\jmath)$ in $SS_\mathfrak{O}(p)$, find an $\mathfrak{O}$-ideal $\mathfrak{a}$ such that $\mathfrak{a} \star (E,\iota) \simeq (E',\iota')$, and an efficient representation of $\varphi_\mathfrak{a} : (F,\jmath) \rightarrow \mathfrak{a}\star(F,\jmath)$.
\end{problem}

Algorithm \ref{alg_classicOvect1orbit} almost solves \Effective{\mathfrak{O}} \textemdash\ it does not give an efficient representation and it only handles the case where $(E,\iota)$ and $(E',\iota')$ are in the same orbit, i.e. when there exists an $\mathfrak{O}$-ideal $\mathfrak{a}$ such that $\mathfrak{a} \star (E,\iota) = (E',\iota')$.
This algorithm follows a meet-in-the-middle approach, namely successive actions of $\mathfrak{O}$-ideals are computed on $(E,\iota)$ and $(E',\iota')$ until a collision is found.

\begin{algorithm}
\caption{Almost \Effective{\mathfrak{O}}}
\begin{flushleft}
\hspace*{\algorithmicindent} \textbf{ Input :}
$(E,\iota),(E',\iota') \in SS_{\mathfrak{O}}(p)$ two efficiently represented oriented elliptic curves in the same orbit and a real $\varepsilon > 0$.\\
\hspace*{\algorithmicindent} \textbf{ Output :}
A $\lceil \log^{2+\varepsilon} |\disc(\mathfrak{O})| \rceil$-smooth $\mathfrak{O}$-ideal with at most $2 \lceil \log | \disc \mathfrak{O} | \rceil$ prime factors which sends $(E,\iota)$ to $(E',\iota')$.
\end{flushleft}
\begin{algorithmic}[1]
	\State $x \gets \lceil \log^{2+\varepsilon} | \disc(\mathfrak{O})| \rceil$. \label{alg_almosteffective_step_x}
	\State $\Sigma_x \gets \{ [\mathfrak{p}] \in \Cl(\mathfrak{O})$, such that  $\gcd(f_\mathfrak{O},\mathfrak{p}) = 1$ and $N(\mathfrak{p}) \leq x$ prime$\}$.
	\State $S_x \gets \Sigma_x \cup \{[\mathfrak{p}]^{-1} \text{ for } [\mathfrak{p}] \in \Sigma_x\}.$
	\State $T[\enc((E,\iota))] \gets (1)$. \label{alg_almosteffective_step_T}
	\While {$\#T < \sqrt{\#\Cl(\mathfrak{O})}$} \label{alg_almosteffective_step_beginwhile1}
		\State $y \gets \Unif\{ y \in \mathbb{N}^{\#S_x} \text{ such that } ||y||_1 = \lceil \log|\disc \mathfrak{O}| \rceil \}$. \label{alg_stepy1}
		\State $\mathfrak{a} \gets S_x^y$.
		\If {$T[\enc(\mathfrak{a} \star (E,\iota))]$ is empty}
			\State $T[\enc(\mathfrak{a} \star (E,\iota))] \gets \mathfrak{a}$. \label{alg_almosteffective_step_endwhile1}
		\EndIf
	\EndWhile
	\State $\mathfrak{a} \gets (1)$.
	\While {$T[\enc(\mathfrak{a} \star (E',\iota'))]$ is empty} \label{alg_almosteffective_step_beginwhile2}
		\State $y \gets \Unif\{ y \in \mathbb{N}^{\#S_x} \text{ such that } ||y||_1 = \lceil \log|\disc \mathfrak{O}| \rceil \}$.
		\label{alg_stepy2}
		\State $\mathfrak{a} \gets S_x^y$.
	\EndWhile
	\State \Return $\bar{\mathfrak{a}} T[\enc(\mathfrak{a}\star(E',\iota'))]$. \label{alg_almosteffective_step_return}
\end{algorithmic}
\label{alg_classicOvect1orbit}
\end{algorithm}

\begin{lemma}[\textbf{GRH}]
\label{lem_classicalOvect}
Algorithm \ref{alg_classicOvect1orbit} runs in expected time $l^{O_\varepsilon(1)}|\disc(\mathfrak{O})|^{1/4}$ where $l$ is the length of the input, and is correct. 
\end{lemma}

\begin{proof}
First of all, notice that using a dictionary structure for the table $T$, one can add and search for elements in time $O(\log \# T)$.
From \cite[Section 5.3]{AC:DDFKLP21}, we have the estimate $\#\Cl(\mathfrak{O}) = O(\log(|\disc(\mathfrak{O})|)\sqrt{|\disc(\mathfrak{O})}|)$.
Then insertions and searches in the table $T$ can be done in $O(\log|\disc(\mathfrak{O})|)$.
Moreover, we use the $\enc$ function, see Section \ref{subsec_encoding}, to have a unique encoding of oriented elliptic curves.
Finally, notice that, for $\disc(\mathfrak{O})$ large enough, we have that $\lceil \log | \disc(\mathfrak{O}) | \rceil \geq C \frac{ \log \# \Cl(\mathfrak{O})}{\log \log | \disc(\mathfrak{O}) |}$, where $C$ is the constant from Proposition \ref{pro_walk}.
Thus Proposition \ref{pro_walk} applies properly to the random walks given by the vectors $y$ such that $||y||_1 = \lceil \log | \disc(\mathfrak{O})| \rceil$ at Step \ref{alg_stepy1} and \ref{alg_stepy2}.

\begin{itemize}
	\item[{[\ref{alg_almosteffective_step_x}-\ref{alg_almosteffective_step_T}]}] Those steps are polynomial in $\log^{2+\varepsilon} \disc(\mathfrak{O})$.
	
	\item[{[\ref{alg_almosteffective_step_beginwhile1}-\ref{alg_almosteffective_step_endwhile1}]}] It is expected that this first while loop will end after $O(\sqrt{\#\Cl(\mathfrak{O})})$ iterations. Indeed, by applying Proposition \ref{pro_walk} to subset of vertices $H := \Cl(\mathfrak{O}) - T$, we get that the probability to get a new element for the table $T$ is greater than $\frac{1}{2} - \frac{1}{\sqrt{\#\Cl(\mathfrak{O})}}$.
	In particular, for $\disc(\mathfrak{O}) \geq 36$, one can expect to add a new element to the table $T$ after at most 3 draws of random smooth ideals.

	By Corollary \ref{cor_smoothactions}, computing an efficient representation of $\mathfrak{a} \star (E,\iota)$ is done in polynomial time in $l_1$, in $\log p$ and in $\log^{2 +\varepsilon}|\disc(\mathfrak{O})|$, where $l_1$ is the length of the representation of $\iota$.

	\item[{[\ref{alg_almosteffective_step_beginwhile2}-\ref{alg_almosteffective_step_return}]}] This while loop is also expected to end after $O(\sqrt{\#\Cl((\mathfrak{O})})$ iterations, since, thanks again to Proposition \ref{pro_walk}, each iteration has a probability of success greater than $\frac{1}{2 \sqrt{\#\Cl(\mathfrak{O})}}$.
	This time, we apply the proposition regarding the landing subset $T$.

	Moreover, as in the first loop, using Corollary \ref{cor_smoothactions}, one can compute the action of $\mathfrak{a}$ in time polynomial in $l_2$, $\log p$ and $\log^{2 + \varepsilon}|\disc(\mathfrak{O})|$, where $l_2$ is the length of the representation of $\iota'$.
\end{itemize}
This leads to a global runtime in $(\max(l_1,l_2) \log p \log^{2+\varepsilon} |\disc(\mathfrak{O})|)^{O(1)}\sqrt{\#\Cl(\mathfrak{O})}$.
Thanks again to the estimate $\#\Cl(\mathfrak{O}) = O(\log(|\disc(\mathfrak{O})|)\sqrt{|\disc(\mathfrak{O})|})$, we get the claimed complexity.\\
	
The correctness of the algorithm is given by a short computation.
By construction, the output $\mathfrak{O}$-ideal $\mathfrak{a}$ verifies
$$
	T[\enc(\mathfrak{a} \star (E',\iota'))] \star (E,\iota) \simeq \mathfrak{a} \star (E',\iota').
$$
Hence,
\begin{align*}
	(\bar{\mathfrak{a}} T [\enc(\mathfrak{a} \star (E',\iota'))]) \star (E,\iota) &=
	\bar{\mathfrak{a}} \star (T [\enc(\mathfrak{a} \star (E',\iota'))] \star (E,\iota))\\
	&\simeq \bar{\mathfrak{a}} \star (\mathfrak{a} \star (E',\iota')) = (\bar{\mathfrak{a}}\mathfrak{a})\star (E',\iota') = (E',\iota).
\end{align*}

Finally, the output ideal is a product of two  $\lceil \log^{2+\varepsilon}|\disc(\mathfrak{O})|\rceil$-smooth $\mathfrak{O}$-ideals with at most $\lceil \log | \disc(\mathfrak{O}) | \rceil$ prime factors thus it is a $\lceil \log^{2+\varepsilon} |\disc(\mathfrak{O})|\rceil$-smooth $\mathfrak{O}$-ideal with at most $2 \lceil \log |\disc(\mathfrak{O})| \rceil$ prime factors.
\end{proof}

\begin{remark}
Algorithm \ref{alg_classicOvect1orbit} needs space exponential in the length of the input.
A space-efficient algorithm is conceivable using a Pollard-$\rho$ approach, as it is used to find isogenies between ordinary elliptic curves in \cite{DCC:BisSut12}.
A rigorous analysis of such algorithms typically requires access to a random oracle, and we do not pursue this direction here.
\end{remark}

Algorithm \ref{alg_classicOvect1orbit} is a central subprocedure in our classical resolution of $\mathfrak{O}$-\Vectorisation\ and $\alpha$-\EndRing, as well as \Effective{\mathfrak{O}}.
These applications of Algorithm \ref{alg_classicOvect1orbit} require to move from one orbit to the other using the $\mathfrak{O}$-twists.

\begin{theorem}[\textbf{GRH}, \Effective{\mathfrak{O}}]
\label{the_classicOvect}
There is a classical algorithm taking as input three oriented elliptic curves $(E,\iota),(E',\iota')$ and $(F,\jmath)$ in $SS_\mathfrak{O}(p)$ and a real number $\varepsilon > 0$ which returns an $\mathfrak{{O}}$-ideal $\mathfrak{a}$ $\lceil \log^{2+\varepsilon}|\disc(\mathfrak{{O}}))| \rceil$-smooth such that $E^\mathfrak{a} \sim E'$ together with a representation of $\varphi_\mathfrak{a}: (F,\jmath) \rightarrow \mathfrak{a} \star(F,\jmath)$ in expected time $l^{O_\varepsilon(1)}|\disc(\mathfrak{O})|^{1/4}$ where $l$ is the length of the input.
The returned representation of $\varphi_\mathfrak{a}$ is given by $O(\log |\disc(\mathfrak{O})|)$ isogeny kernels of order at most $\lceil \log^{2+\varepsilon} | \disc(\mathfrak{O})| \rceil$. 
\end{theorem}

\begin{proof}
Suppose we are given some positive real $\varepsilon$ and two oriented supersingular elliptic curves $(E,\iota) \not \simeq (E',\iota') \in SS_\mathfrak{O}(p)$, where $\mathfrak{O}$ is an order of some quadratic field $K$.
First, we check if $p$ is inert or ramified in $K$.
Recall that $p$ does not split over $K$ otherwise $SS_\mathfrak{O}(p)$ would be empty \cite[Proposition 3.2]{onuki_oriented_2021}.

By \cite[Theorem 4.4]{arpin_orientations_2024}, if $p$ is ramified in $K$, then the action of $\Cl(\mathfrak{O})$ has only one orbit.
Thus by running Algorithm \ref{alg_classicOvect1orbit} with the inputs $(E,\iota),(E',\iota')$ and $\varepsilon$, we get an $\mathfrak{O}$-ideal $\mathfrak{a}$ such that $\mathfrak{a} \star (E,\iota) \simeq (E',\iota')$.

Otherwise, if $p$ is inert in $K$, again by \cite[Theorem 4.4]{arpin_orientations_2024}, the action of $\Cl(\mathfrak{O})$ has two orbits.
We then run two instances of Algorithm \ref{alg_classicOvect1orbit} in parallel, the first one with the inputs $(E,\iota),(E',\iota')$ and $\varepsilon$ and the second one with the inputs $(E,\bar{\iota}),(E',\iota')$ and $\varepsilon$.
We know that $(E,\iota)$ and its $\mathfrak{O}$-twist $(E,\bar{\iota})$ are not in the same orbit, see \cite{onuki_oriented_2021}, thus only one procedure will stop.
If it is the instance having $(E,\iota)$ as input, that means that we find an $\mathfrak{O}$-ideal $\mathfrak{a}$ sending $(E,\iota)$ to $(E',\iota')$.
Else, it means that $(E,\bar{\iota})$ and $(E',\iota')$ are in the same orbit. Hence $(E,\iota)$ is not on the same orbit as $(E',\iota')$ and there is no solution to the \Effective{\mathfrak{O}} problem.
In this case, we return \texttt{False}.

Now we have an ideal $\mathfrak{a}$ solving our \Effective{\mathfrak{O}} instance, it remains to compute an efficient representation of the isogeny $\varphi_\mathfrak{a}$.
Since $\mathfrak{a}$ has been returned by Algorithm \ref{alg_classicOvect1orbit} it is a $\lceil \log^{2+\varepsilon}|\disc(\mathfrak{O})|\rceil$-smooth $\mathfrak{O}$-ideal with at most $2 \log |\disc(\mathfrak{O})|$ prime factors.
Then using Corollary \ref{cor_smoothactions} an efficient representation of $\varphi_\mathfrak{a}$ can be computed in time polynomial in $\log p$, $\lceil \log^{2 + \varepsilon} |\disc(\mathfrak{O})| \rceil$ and in the length of the representation of $\iota$. 
\end{proof}

\begin{theorem}[\textbf{GRH}, Classical $\mathfrak{O}$-\Vectorisation]
\label{the_fullclassicOvect}
There is a classical algorithm taking as input two oriented elliptic curves $(E,\iota)$ and $(E',\iota')$ in $SS_\mathfrak{O}(p)$ and a real number $\varepsilon > 0$ which returns an $\mathfrak{O}$-ideal $\mathfrak{a}$ of $\lceil \log^{2+\varepsilon} |\disc(\mathfrak{O})| \rceil$-smooth norm such that $E^\mathfrak{a} \sim E'$ in expected time $l^{O\varepsilon(1)} |\disc (\mathfrak{O})|^{1/4}$ where $l$ is the length of the input. 
\end{theorem}

\begin{proof}
We know that the action of $\Cl(\mathfrak{O})$ on $SS_\mathfrak{O}(p)$ has at most $2$ orbits, see Proposition \ref{pro_groupaction}.
Let $O$ be the orbit of $(E',\iota')$. 
From the same proposition, we know that $(E,\iota)$ or its $\mathfrak{O}$-twist $(E,\bar{\iota})$ is in $O$. 
Thus, by running two instances of Algorithm \ref{alg_classicOvect1orbit} until one ends, the first taking as input the oriented elliptic curves $(E,\iota)$ and $(E',\iota')$ and the second taking $(E,\bar{\iota})$ and $(E',\iota')$, we make sure that we find a suitable ideal in an expected time given by Lemma \ref{lem_classicalOvect}.
\end{proof}

From the previous results, we can now prove Theorem \ref{the_classical}.

\begin{proof}[Proof of Theorem \ref{the_classical}]
Let $E$ be a primitively $\mathfrak{O}$-orientable elliptic curve defined over $\bar{\mathbb{F}}_p$ and $\iota : \mathbb{Z}[\alpha] \hookrightarrow \End(E)$ be an orientation of $E$ such that $\mathbb{Z}[\alpha] \subseteq \mathfrak{O}$.
Let us prove that the computation of the endomorphism ring $\End(E)$ can be done in probabilistic time $l^{O(1)} |\disc(\mathbb{Z}[\alpha])|^{1/4}$, where $l$ is the length of the input.

First we compute the factorisation of $\disc(\mathbb{Z}[\alpha])$ in time subexponential in the length of $\disc(\mathbb{Z}[\alpha])$, see for instance \cite{pomerance_fast_1987}.
Then, by Corollary \ref{cor_primitivisation}, we can compute, in probabilistic time polynomial in the length of the input, the primitive orientation $\jmath$ such that $(E,\jmath) \in SS_\mathfrak{O}(p)$.
This reduces the computation of $\End(E)$ to the instance of $\mathfrak{O}$-\EndRing\ given by $(E,\jmath)$ which, in turn, reduces in probabilistic polynomial time to an instance of $\mathfrak{O}$-\Vectorisation\ by Proposition \ref{pro_OEndRingreduction}.
Finally, by Theorem \ref{the_fullclassicOvect} and since $|\disc(\mathbb{Z}[\alpha])|$ is greater than $|\disc(\mathfrak{O})|$, the $\mathfrak{O}$-\Vectorisation\ problem can be solved in $l^{O(1)} |\disc(\mathbb{Z}[\alpha])|^{1/4}$.
\end{proof}

\end{subsection}

\begin{subsection}{Quantum algorithm}

The subexponential quantum resolution of the $\mathfrak{O}$-\Vectorisation\  proven in this section is based on the work of Childs, Jao and Soukharev to construct an isogeny between two given isogenous ordinary elliptic curves, \cite{childs_constructing_2014}.
In particular, we use the fact that given two oriented elliptic curves $(E_0,\iota_0),(E_1,\iota_1) \in SS_\mathfrak{O}(p)$ in the same orbit, finding an $\mathfrak{O}$-ideal $\mathfrak{a}$ such that $\mathfrak{a} \star (E_0,\iota_0) = (E_1,\iota_1)$ can be viewed as an instance of the \HiddenShift\ problem.

\begin{problem}[\HiddenShift]
\label{prob_hiddenshift}
Given a finite abelian group $(A,+)$, a finite set $S \subset \{0,1\}^m$ of encoding length $m$ and two black-box functions $f_0,f_1 : A \rightarrow S$ where $f_0$ is injective and such that there exists an element $s \in A$ verifying $f_1(x) = f_0(s + x)$ for any $x \in A$, find the element $s$ called the shift hidden by $f_0$ and $f_1$.
\end{problem}

In this paper, we assume that the abelian group $A$ of any instance of \texttt{Hidden shift} is always given as $\mathbb{Z}/n_1\mathbb{Z} \times \dots \times \mathbb{Z}/n_k\mathbb{Z}$ for some integers $k,n_1,\dots n_k$.
Notice that the \HiddenShift\ problem can also be considered when $A$ is not abelian.
Nevertheless the above formulation of the problem allows us to use Kuperberg's quantum algorithm to solve it in a subexponential number of queries of the black-box functions $f_0$ and $f_1$.

\begin{theorem}[Theorem 7.1. \cite{kuperberg_subexponential-time_2005}]
\label{the_kuperberg}
There is a quantum algorithm such that the \texttt{Hidden Shift problem} for abelian groups can be solved with time and query complexity $2^{O(\sqrt{\log n})}$, where $n$ is the size of the abelian group, uniformly for all finitely generated abelian groups.
\end{theorem}

To solve quantumly $\mathfrak{O}$-\Vectorisation, we first prove the correctness and the expected subexponential runtime of Algorithm \ref{alg_quantumOvect} which solves $\mathfrak{O}$-\Vectorisation\ assuming that the two input curves are in the same orbit.
This algorithm is analogous to \cite[Algorithm 3]{childs_constructing_2014}.

\begin{algorithm}
\caption{Quantum $\mathfrak{O}$-\Vectorisation\  in the same orbit}
\begin{flushleft}
\hspace*{\algorithmicindent} \textbf{ Input :}
$(E_0,\iota_0),(E_1,\iota_1) \in SS_{\mathfrak{O}}(p)$ two oriented elliptic curves in the same orbit.\\
\hspace*{\algorithmicindent} \textbf{ Output :}
a reduced $\mathfrak{O}$-ideal $\mathfrak{a} \in \Cl(\mathfrak{O})$ such that $\mathfrak{a} \star (E_0,\iota_0) = (E_1,\iota_1)$ and the isogeny $\varphi_\mathfrak{a}: (E_0,\iota_0) \rightarrow (E_1,\iota_1)$. 
\end{flushleft}

\begin{algorithmic}[1]
	\State Compute $\Cl(\mathfrak{O})$ as a decomposition $\langle [\mathfrak{b}_1] \rangle \oplus \dots \oplus \langle [\mathfrak{b}_k] \rangle$. \label{alg_quantumvect_classgroup} 
	\State Denote by $n_j$ the order of $\langle [\mathfrak{b}_j] \rangle$, for $j \in \llbracket{1,\dots,k\rrbracket}$.
	\State Solve the \texttt{Hidden Shift} problem instance given with the black-box functions, for $j \in \{0,1\}$, $$f_j : \mathbb{Z}/n_1\mathbb{Z} \times \dots \times \mathbb{Z}/n_k\mathbb{Z} \rightarrow \enc(SS_\mathfrak{O}(p)), (x_1,\dots,x_k) \mapsto \enc((\mathfrak{b}_1^{x_1}\dots\mathfrak{b}_k^{x_k}) \star (E_j,\iota_j))$$ where $s = (s_1,\dots,s_k)$ denoted the hidden shift. \label{alg_quantumvect_hiddenshift}
	\State Compute $\mathfrak{a}$ the reduced representative of the ideal class $[\mathfrak{b}_1^{s_1}\dots\mathfrak{b}_k^{s_k}]$. \label{alg_quantumvect_reducedrep}
	\State Compute the isogeny $\varphi_\mathfrak{a}$ induced by the ideal $\mathfrak{a}$. \label{alg_quantumvect_isogeny}
	\State \Return $\mathfrak{a}$ and $\varphi_\mathfrak{a}$.
\end{algorithmic}
\label{alg_quantumOvect}
\end{algorithm}

\begin{lemma}[\textbf{GRH}]
\label{lem_quantum_aux}
The Algorithm \ref{alg_quantumOvect} is correct and runs in $l^{O(1)}  L_{|\disc(\mathfrak{O})|}[1/2]$ expected time where $l$ is the length of the input.
\end{lemma}
\begin{proof}
Let us prove the complexity of Algorithm \ref{alg_quantumOvect}:
\begin{itemize}
	\item[{[\ref{alg_quantumvect_classgroup}]}] Under GRH, one can quantumly compute the group structure of $\Cl(\mathfrak{O})$ in time polynomial in $\log |\disc(\mathfrak{O})|$, using for instance \cite[Theorem 1.2]{biasse_efficient_2016}.

	\item[{[\ref{alg_quantumvect_hiddenshift}]}] By Kuperberg's algorithm, Theorem \ref{the_kuperberg}, one can solve the instance of the \texttt{Hidden Shift} problem with time and query complexity $L_{|\disc(\mathfrak{O})|}[1/2]$. All the queries are done on the function $f_0$ and $f_1$ which can be evaluated in time polynomial in the length of the input using \texttt{Clapoti}, see Corollary \ref{cor_clapoti}.
	Thus this step in done in $l^{O(1)}L_{|\disc(\mathfrak{O})|}[1/2]$, where $l$ is the length of the input.

	\item[{[\ref{alg_quantumvect_reducedrep}]}] To compute the reduced representative of the ideal class $[\mathfrak{b}_1^{s_1}\dots\mathfrak{b}_k^{s_k}]$, we use a square-and-multiply approach where the ideal computed at each step is reduced.
	With this method, $\forall i \in \llbracket 1,k\rrbracket$, $[\mathfrak{b}_i^{s_i}]$ can be reduced in $O(\lceil \log \#\Cl(\mathfrak{O})\rceil)$ squarings, multiplications and reductions which all can be done in polynomial time in $\log |\disc \mathfrak{O}|$.
	Then it only remains to compute the reduced representative of $[\mathfrak{b}_1^{s_1}\dots\mathfrak{b}_k^{s_k}]$ from the reduced representatives of $[\mathfrak{b}_1^{s_1}]$,\dots,$[\mathfrak{b}_k^{s_k}]$ in time polynomial in $\log |\disc \mathfrak{O}|$.
	Hence, using the standard result $\# \Cl(\mathfrak{O}) = O(\log(|\disc(\mathfrak{O}))\sqrt{\disc(\mathfrak{O})})$, this whole step is done in time polynomial in $\log |\disc(\mathfrak{O})|$.

\item[{[\ref{alg_quantumvect_isogeny}]}] Finally with \texttt{Clapoti}, we can compute the isogeny $\varphi_\mathfrak{a} : (E_0,\iota_0) \rightarrow (E_1,\iota_1)$ in time polynomial in the length of the input. 
\end{itemize}

A short computation proves that the shift $s = (s_1,\dots,s_k)$ hidden by $f_0$ and $f_1$ gives the ideal class $[\mathfrak{a}] = [\mathfrak{b}_1^{s_1}\dots \mathfrak{b}_k^{s_k}]$ such that $\mathfrak{a} \star (E_0,\iota_0) = (E_1,\iota_1)$.
Indeed, for every $[\mathfrak{b}] \in \Cl(\mathfrak{O})$, there is a vector $b = (b_1,\dots,b_k) \in \mathbb{Z}/n_1\mathbb{Z} \times \dots \times \mathbb{Z}/n_k\mathbb{Z}$ such that $[\mathfrak{b}] = [\mathfrak{b}_1^{b_1}\dots\mathfrak{b}_k^{b_k}]$.
Then,
\begin{align*}
	f_1(b) &= \enc((\mathfrak{b}_1^{b_1}\dots\mathfrak{b}_k^{b_k}) \star (E_1,\iota_1)) = \enc(\mathfrak{b} \star (E_1,\iota_1))\\
	&= \enc((\mathfrak{b} \mathfrak{a}) \star (E_0,\iota_0)) = \enc((\mathfrak{b}_1^{a_1+b_1}\dots\mathfrak{b}_k^{a_k + b_k}) \star (E_0,\iota_0))\\
	&= f_0(a + b).
\end{align*}
Finally, the \HiddenShift\ problem is well defined as $f_0$ is injective because the action of $\Cl(\mathfrak{O})$ over $SS_\mathfrak{O}(p)$ is free.\\
\end{proof}

\begin{theorem}[\textbf{GRH}, Quantum $\mathfrak{O}$-\Vectorisation]
\label{the_quantumOvect}
There is a quantum algorithm taking as input two oriented elliptic curves $(E_0,\iota_0)$ and $(E_1,\iota_1)$ in $SS_\mathfrak{O}(p)$ which returns an $\mathfrak{O}$-ideal $\mathfrak{a}$ such that $E_0^{\mathfrak{a}} \sim E_1$ together with the associated isogeny $\varphi_\mathfrak{a} : E_0 \rightarrow E_1$.
This algorithm runs in expected time $l^{O(1)}L_{|\disc(\mathfrak{O})|}[1/2]$ where $l$ is the length of the input.
\end{theorem}

\begin{proof}
As for the classical resolution of $\mathfrak{O}$-\Vectorisation, it is sufficient to run two instances of Algorithm \ref{alg_quantumOvect}.
The first one with the inputs $(E_0,\iota_0)$ and $(E_1,\iota_1)$ and the second one with the inputs $(E_0,\bar{\iota_0})$ and $(E_1,\iota_1)$.
When the action has 2 orbits (see Proposition \ref{pro_groupaction}), one of these instances corresponds to a scenario where there is no solution to the \HiddenShift\ problem.
It can be assumed that running Algorithm \ref{alg_quantumOvect} on this incorrect instance will either yield an erroneous output or fail to terminate.
However, since we can check the validity of a solution in polynomial time using with \texttt{Clapoti}, this does not pose any issues.
Consequently, the complexity in the Theorem \ref{the_quantumOvect} directly follows from Lemma \ref{lem_quantum_aux}.
\end{proof}

This leads us to the following proof of Theorem \ref{the_quantum}.

\begin{proof}[Proof of Theorem \ref{the_quantum}]
\label{proof_quantum}
By Corollary \ref{cor_primitivisation} and because the factorisation of $\disc(\mathbb{Z}[\alpha])$ can be computed in quantum polynomial time, the $\alpha$-\EndRing\ problem reduces to $\mathfrak{O}$-\EndRing\ in time polynomial in the length of the instance.
Notice that the discriminant of the order returned by this primitivisation step can only decrease in absolute value. 
Then, by Proposition \ref{pro_OEndRingreduction}, $\alpha$-\EndRing\ reduces to $\mathfrak{O}$-\Vectorisation\ in probabilistic time polynomial in the length of the input.
Hence, by Theorem \ref{the_quantumOvect}, $\alpha$-\EndRing\ can be solved in expected time $l^{O(1)}L_{|\disc{\mathbb{Z}[\alpha]}|}[1/2]$.
\end{proof}

\end{subsection}
\end{section}
\begin{section}{Ascending the volcano} 
\label{sec_navigate}
We fix $K$ to be a quadratic number field and we consider supersingular elliptic curves over the finite field $\mathbb{F}_{p^2}$ where $p$ is a prime which does not split in $K$.
Let $\ell \not = p$ be a prime number.\\

Adding $K$-orientations to an $\ell$-isogeny graph of supersingular elliptic curves gives a structure of volcano to each of its connected components, analogous to the structure of isogeny graphs of ordinary elliptic curves. 
We now introduce formally this notion before showing how results of Section \ref{sec_higher} can be used to naviguate efficiently in this volcano.
This is then used to optimise results of the previous section.
Notably, we improve \cite[Theorem 5]{EC:Wesolowski22} by proving that $(\mathbb{Z} + c\mathfrak{O})$-\EndRing\ reduces to $\mathfrak{O}$-\EndRing\ in polynomial time in the largest prime factor of $c$ (instead of it largest \emph{prime power} factor).

We define the \textbf{$K$-oriented $\ell$-isogeny graphs} as the graph having for set of vertices the $K$-oriented supersingular elliptic curves up to $K$-isomorphism and for edges the $K$-oriented isogenies of degree $\ell$ between them.

Let $(E,\iota),(E',\iota')$ be two $K$-oriented supersingular elliptic curves, where $\iota$ is a primitive $\mathfrak{O}$-orientation and $\iota'$ is a primitive $\mathfrak{O}'$-orientation.

For any $K$-oriented isogeny $\varphi: (E,\iota) \rightarrow (E',\iota')$ of degree $\ell$, we say that $\varphi$ is
\begin{center}
\begin{itemize}
	\item[$\nearrow$] \begin{center} \textbf{ascending} if $\mathfrak{O} \varsubsetneq \mathfrak{O}'$, \end{center}
	\item[$\rightarrow$] \begin{center} \textbf{horizontal} if $\mathfrak{O} = \mathfrak{O}'$, \end{center}
	\item[$\searrow$] \begin{center}  \textbf{descending} if $\mathfrak{O} \varsupsetneq \mathfrak{O}'$. \end{center}
\end{itemize}
\end{center}

We denote by $\big( \frac{\disc(\mathfrak{O})}{\ell} \big)$ the Legendre symbol.
From \cite{colo_orienting_2020}, the oriented elliptic curve $(E,\iota)$ has $\ell - \big( \frac{\disc(\mathfrak{O})}{\ell} \big)$ descending isogenies from it.
Moreover, there are in addition
\begin{itemize}
	\item $\big( \frac{\disc(\mathfrak{O})}{\ell} \big) + 1$ horizontal isogenies, if $\mathfrak{O}$ is maximal at $\ell$,
	\item one ascending isogeny, otherwise.
\end{itemize}
Moreover, an isogeny between $K$-oriented elliptic curves of non-prime degree is said to be ascending, horizontal or descending if its factorisation into prime-degree isogenies is only composed of ascending, horizontal or descending isogenies.

Then, we say that each component of the $K$-oriented $\ell$-isogeny graph has a volcano structure as its shape recalls one.
Indeed, it has a finite cycle of horizontal isogenies, called the \textbf{crater}, the surface or the rim, such that from each vertex starts an infinite tree of vertical isogenies.
In particular, an oriented elliptic curve $(E,\iota) \in SS_\mathfrak{O}(p)$ is at the crater of the $K$-oriented $\ell$-isogeny graph if and only if $\mathfrak{O}$ is maximal at $\ell$.
Otherwise, we say that $(E,\iota)$ is at \textbf{depth} $m$ if the valuation at $\ell$ of $[O_K:\mathfrak{O}]$ is $m$, where $O_K$ is the maximal order of $K$.
This means that one can walk from $(E,\iota)$ to the crater of the $K$-oriented $\ell$-isogeny graph by taking $m$ ascending steps.\\

We provide an algorithm to walk to the crater of the volcano as an example of efficient navigation.

\begin{lemma}[Walking to the crater]
\label{lem_walktotherim}
Let $(E,\iota) \in SS_\mathfrak{O}(p)$ be a $\mathfrak{O}$-oriented elliptic curve and $\ell \not = p$ a prime number.
If $(E,\iota)$ is at depth at least $m$ in the $K$-oriented $\ell$-isogeny volcano, then one can compute the unique ascending isogeny $\varphi: (E,\iota) \rightarrow (E',\iota')$ of degree $\ell^m$ in time polynomial in $\ell, m, \log p$ and in the length of the representation of $\iota$.

In particular, one can give the representation of $\varphi$ as $m$ kernels of successive isogenies all defined over an extension of degree $O(\ell^2)$. 
\end{lemma}

\begin{proof}
Let $(E,\iota) \in SS_\mathfrak{O}(p)$ be an $\mathfrak{O}$-oriented elliptic curve at depth $m$ in the $K$-oriented $\ell$-isogeny volcano and $\varphi: (E,\iota) \rightarrow (E',\iota')$ be the unique ascending $K$-isogeny of degree $\ell^m$.
We compute the isogeny $\varphi$ by composing the $m$ successive ascending isogenies of degree $\ell$ from $(E,\iota)$.

Let $\varphi_1 :(E,\iota) \rightarrow (E_1,\iota_1)$ be the unique ascending isogeny of degree $\ell$ from $(E,\iota)$.
We denote by $\mathfrak{O}_1$ the order such that $(E_1,\iota_1)$ is $\mathfrak{O}_1$-primitively oriented and $\mathfrak{O}$ is a suborder of $\mathfrak{O}_1$.
Let $\omega_1$ be a generator of $\mathfrak{O}_1$.
We assume, without loss of generality, that $\mathfrak{O}$ is given by a generator $\omega$ of the form $\omega = \ell \omega_1$.
Then as shown in \cite[Lemma 11]{EC:Wesolowski22}, $\ker \varphi = \ker(\iota(\omega))\cap E[\ell]$.
As $\iota(\omega)$ is efficiently represented, $\ker \varphi$ can be computed in time polynomial in $\ell,\log p$ and $l_0$, where $l_0$ is the length of the representation of $\iota$.
One just has to compute a basis of $E[\ell]$ and to take a generator of the cyclic subgroup of $E[\ell]$ that vanishes under $\iota(\omega)$.
It provides a representation of $\varphi$ given by its kernel generated by a point living in an extension of degree $O(\ell^2)$.
Thus, it is possible to compute the elliptic curve $E_1 = E/\ker \varphi$ and its orientation $\iota_1$ induced by $\varphi_1$ in time polynomial in $\ell,\log p$ and in $l_0$.

On the one hand, to recover $E_1$, we use Vélu's formula \cite{velu_isogenies_1971}.
On the other hand, for the computation of the induced orientation, we have
$$
\iota_1(\omega_1) = {\varphi_1}_*(\iota(\omega_1)) = \frac{\varphi \circ \iota(\omega_1) \circ \hat\varphi}{\ell} = \frac{\varphi \circ \iota(\ell \omega_1) \circ \hat\varphi}{\ell^2} = \frac{\varphi \circ \iota(\omega) \circ \hat\varphi}{\ell^2}.
$$
Thus, from the known representations of $\varphi$ and $\iota(\omega)$, we get an efficient representation of $\varphi \circ \iota(\omega) \circ \hat\varphi$ and we just need to divide it by $\ell^2$ using Algorithm \ref{alg_division}.
By Theorem \ref{cor_division}, this computation is polynomial in $l, \log p$ and in $l_0$ and returns a representation of $\iota_1$ of size $O(\log(p)\log^3(\ell^2l_0))$ such that one can evaluate it on a point in $\tilde{O}(\log^{11}(\ell^2 l_0))$ operations over its field of definition.\\

We do the same computation to get a representation of the unique ascending isogeny $\varphi_2 : (E_1,\iota_1) \rightarrow (E_2,\iota_2)$ of degree $\ell$.
First, we compute the kernel $\ker \varphi_2 = \ker(\iota_1(\omega_1)) \cap E_1[\ell]$ and deduce the curve $E_2 = E_1/\ker \varphi_2$ together with a representation of $\varphi_2$ in time polynomial in $\ell, \log p$ and in $l_0$.
Then we recover in time polynomial in $\ell, \log p$ and $l_0$ a representation of the induced orientation $\iota_2$, with the same properties as the one of $\iota_1$.\\

After such $m$ steps, one can provide efficient representations for the totality of the $\varphi_i$, for $i \in \llbracket 1,m \rrbracket$, in polynomial time in $\ell, \log p, 
l_0$ and $m$.
The representation $\varphi : (E,\iota) \rightarrow (E',\iota')$ is then given by the composition of the representations of $\varphi_i$, for $i \in \llbracket 1,m \rrbracket$.
Hence, this representation is provided by the kernels of the $m$ successive isogenies, namely by $m$ points living in extension of degree $O(\ell^2)$.
\end{proof}

\begin{theorem}[\textbf{GRH}]
\label{the_cOtoO}
Let $c$ be a positive integer and $\mathfrak{O}$ a quadratic order.
Then $(\mathbb{Z}+c \mathfrak{O})$-\EndRing\ reduces to $\mathfrak{O}$-\EndRing\ in probabilistic polynomial time in the length of the input and in the largest prime factor of $c$. 
\end{theorem}

\begin{proof}
Let $(E,\iota) \in SS_{\mathbb{Z} + c\mathfrak{O}}(p)$ be an instance of $(\mathbb{Z} + c\mathfrak{O})$-\EndRing.
Let us solve it using an $\mathfrak{O}$-\EndRing\ oracle.

Here, the main objective is to compute a representation of the unique isogeny $\varphi : E \rightarrow E'$ of degree $c$ such that $\varphi_*(\iota)$ is an $\mathfrak{O}$-orientation.
Indeed, using the $\mathfrak{O}$-\EndRing\ oracle on the instance $(E',\varphi_*(\iota))$ gives an $\varepsilon$-basis of $\End(E')$.
Then, from the $\varepsilon$-basis of $\End(E')$ and $\hat\varphi$, an $\varepsilon$-basis of $\End(E)$ can be computed, under GRH, in probabilistic polynomial time in the length of the input, \cite[Lemma 12]{EC:Wesolowski22}.
Notice that to use directly \cite[Lemma 12]{EC:Wesolowski22}, the isogeny $\hat\varphi$ needs to be represented by its kernel.
It is not an issue for this proof.\\

First, we compute the prime factorisation of $c$ and denote it $\prod_{i=1}^r \ell_i^{e_i}$.
This factorisation can be done in polynomial time in $P^+(c)$.
Using Lemma \ref{lem_walktotherim}, we can successively take $e_i$ steps to the crater of the oriented $\ell_i$-isogeny volcanoes, for $i \in \llbracket 1,r \rrbracket$, to reach $(E',\varphi_*(\iota))$ in polynomial time in the length of the input and in $P^+(c)$.
Let us denote by $(E_i,\iota_i)$ the oriented elliptic curve obtained by walking $e_1$ steps from $(E_0,\iota_0) := (E,\iota)$ to the crater of the oriented $\ell_1$-isogeny volcano then $e_2$ steps to the crater of the oriented $\ell_2$-isogeny volcano and so on until walking $e_i$ steps to the crater of the oriented $\ell_i$-isogeny volcano.
We denote by $\varphi_i$ the isogeny of degree $\ell_i^{e_i}$ that maps $(E_{i-1},\iota_{i-1})$ to $(E_i,\iota_i)$.
By Lemma \ref{lem_walktotherim}, every $\varphi_i$ is given by $e_i$ successive kernels of $\ell_i$-isogenies living in extension of degree $O({P^+(c)}^2)$.
We then denote this decomposition of $\varphi_i$ into $\ell_i$ isogenies by $\varphi_i = \phi_{i,e_i} \circ \dots \circ \phi_{i,1}$.
Finally, using the decomposition of every $\varphi_i$ into isogenies of prime degree, we have the following decomposition of $\hat\varphi : (E',\iota') \rightarrow (E,\iota)$
$$
\hat\varphi = \hat\phi_{1,1} \circ \dots \circ \hat\phi_{1,e_1} \circ \hat\phi_{2,1} \circ \dots \circ \hat\phi_{2,e_2} \circ \dots \circ \hat\phi_{r,1} \circ \dots \circ \hat\phi_{r,e_r},
$$
where all the kernels of the $\hat\phi_{i,j}$ are recoverable in time polynomial in $P^+(c)$ and in $\log p$.

Finally, $\End(E)$ is computable in probabilistic polynomial time in the length of the input and in $P^+(c)$ by propagating the knowledge of the endomorphism ring from $(E',\iota')$ to $(E,\iota)$ using the $O(\log c)$ dual isogenies of prime degree between them, thanks to \cite[Lemma 12]{EC:Wesolowski22}.
\end{proof}

\begin{corollary}[\textbf{GRH}]
Let $c$ be a positive integer and $\mathfrak{O}$ a quadratic order.
Then $(\mathbb{Z}+c\mathfrak{O})$-\EndRing\ can be solved in probabilistic polynomial time in $(l \cdot P^+(c))^{O(1)}|\disc(\mathfrak{O})|^{1/4}$ where $l$ is the length of the input and $P^+(c)$ is the largest prime factor of $c$.
\end{corollary}

\begin{proof}
This is a direct consequence of the reduction of $(\mathbb{Z} + c\mathfrak{O})$-\EndRing\ to $\mathfrak{O}$-\EndRing\ given by Theorem \ref{the_cOtoO} together with the complexity result on $\mathfrak{O}$-\EndRing\ given by Theorem \ref{the_classical}.
\end{proof}

\end{section}

\bibliography{ref,abbrev1,crypto}

\newcommand{\etalchar}[1]{$^{#1}$}
\begin{thebibliography}{MMP{\etalchar{+}}23}

\bibitem[ACL{\etalchar{+}}23]{arpin_orienteering_2023}
Sarah Arpin, Mingjie Chen, Kristin~E. Lauter, Renate Scheidler, Katherine~E.
  Stange, and Ha~T.~N. Tran.
\newblock Orienteering with one endomorphism.
\newblock {\em Matematica}, 2(3):523--582, 2023.
\newblock \href {https://doi.org/10.1007/s44007-023-00053-2}
  {\path{doi:10.1007/s44007-023-00053-2}}.

\bibitem[ACL{\etalchar{+}}24]{arpin_orientations_2024}
Sarah Arpin, Mingjie Chen, Kristin~E. Lauter, Renate Scheidler, Katherine~E.
  Stange, and Ha~T.~N. Tran.
\newblock Orientations and cycles in supersingular isogeny graphs.
\newblock In {\em Research directions in number theory}, volume~33 of {\em
  Assoc. Women Math. Ser.}, pages 25--86. Springer, Cham, 2024.
\newblock \href {https://doi.org/10.1007/978-3-031-51677-1_2}
  {\path{doi:10.1007/978-3-031-51677-1_2}}.

\bibitem[BJS14]{INDOCRYPT:BiaJaoSan14}
Jean-Fran{\c c}ois Biasse, David Jao, and Anirudh Sankar.
\newblock A quantum algorithm for computing isogenies between supersingular
  elliptic curves.
\newblock In Willi Meier and Debdeep Mukhopadhyay, editors, {\em Progress in
  Cryptology - INDOCRYPT~2014: 15th International Conference in Cryptology in
  India}, volume 8885 of {\em Lecture Notes in Computer Science}, pages
  428--442. Springer, Cham, December 2014.
\newblock \href {https://doi.org/10.1007/978-3-319-13039-2_25}
  {\path{doi:10.1007/978-3-319-13039-2_25}}.

\bibitem[BKV19]{AC:BeuKleVer19}
Ward Beullens, Thorsten Kleinjung, and Frederik Vercauteren.
\newblock {CSI}-{FiSh}: Efficient isogeny based signatures through class group
  computations.
\newblock In Steven~D. Galbraith and Shiho Moriai, editors, {\em Advances in
  Cryptology -- {ASIACRYPT}~2019, Part~I}, volume 11921 of {\em Lecture Notes
  in Computer Science}, pages 227--247. Springer, Cham, December 2019.
\newblock \href {https://doi.org/10.1007/978-3-030-34578-5_9}
  {\path{doi:10.1007/978-3-030-34578-5_9}}.

\bibitem[BS12]{DCC:BisSut12}
Gaetan Bisson and Andrew~V. Sutherland.
\newblock A low-memory algorithm for finding short product representations in
  finite groups.
\newblock {\em Designs, Codes and Cryptography}, 63(1):1--13, 2012.
\newblock \href {https://doi.org/10.1007/s10623-011-9527-8}
  {\path{doi:10.1007/s10623-011-9527-8}}.

\bibitem[BS16]{biasse_efficient_2016}
Jean-François Biasse and Fang Song.
\newblock Efficient quantum algorithms for computing class groups and solving
  the principal ideal problem in arbitrary degree number fields.
\newblock In {\em Proceedings of the {T}wenty-{S}eventh {A}nnual {ACM}-{SIAM}
  {S}ymposium on {D}iscrete {A}lgorithms}, pages 893--902. ACM, New York, 2016.
\newblock \href {https://doi.org/10.1137/1.9781611974331.ch64}
  {\path{doi:10.1137/1.9781611974331.ch64}}.

\bibitem[CD20]{PQCRYPTO:CasDec20}
Wouter Castryck and Thomas Decru.
\newblock {CSIDH} on the surface.
\newblock In Jintai Ding and Jean-Pierre Tillich, editors, {\em Post-Quantum
  Cryptography - 11th International Conference, PQCrypto 2020}, pages 111--129.
  Springer, Cham, 2020.
\newblock \href {https://doi.org/10.1007/978-3-030-44223-1_7}
  {\path{doi:10.1007/978-3-030-44223-1_7}}.

\bibitem[CD23]{EC:CasDec23}
Wouter Castryck and Thomas Decru.
\newblock An efficient key recovery attack on {SIDH}.
\newblock In Carmit Hazay and Martijn Stam, editors, {\em Advances in
  Cryptology -- {EUROCRYPT}~2023, Part~V}, volume 14008 of {\em Lecture Notes
  in Computer Science}, pages 423--447. Springer, Cham, April 2023.
\newblock \href {https://doi.org/10.1007/978-3-031-30589-4_15}
  {\path{doi:10.1007/978-3-031-30589-4_15}}.

\bibitem[CJS14]{childs_constructing_2014}
Andrew Childs, David Jao, and Vladimir Soukharev.
\newblock Constructing elliptic curve isogenies in quantum subexponential time.
\newblock {\em Journal of Mathematical Cryptology}, 8(1):1--29, 2014.
\newblock \href {https://doi.org/10.1515/jmc-2012-0016}
  {\path{doi:10.1515/jmc-2012-0016}}.

\bibitem[CK20]{colo_orienting_2020}
Leonardo Colo and David Kohel.
\newblock Orienting supersingular isogeny graphs.
\newblock {\em Journal of Mathematical Cryptology}, 14(1):414--437, 2020.
\newblock Publisher: De Gruyter.
\newblock \href {https://doi.org/10.1515/jmc-2019-0034}
  {\path{doi:10.1515/jmc-2019-0034}}.

\bibitem[CLG09]{charles_cryptographic_2009}
Denis~X Charles, Kristin~E Lauter, and Eyal~Z Goren.
\newblock Cryptographic hash functions from expander graphs.
\newblock {\em Journal of CRYPTOLOGY}, 22(1):93--113, 2009.
\newblock Publisher: Springer.
\newblock \href {https://doi.org/10.1007/s00145-007-9002-x}
  {\path{doi:10.1007/s00145-007-9002-x}}.

\bibitem[CLM{\etalchar{+}}18]{AC:CLMPR18}
Wouter Castryck, Tanja Lange, Chloe Martindale, Lorenz Panny, and Joost Renes.
\newblock {CSIDH}: An efficient post-quantum commutative group action.
\newblock In Thomas Peyrin and Steven Galbraith, editors, {\em Advances in
  Cryptology -- {ASIACRYPT}~2018, Part~III}, volume 11274 of {\em Lecture Notes
  in Computer Science}, pages 395--427. Springer, Cham, December 2018.
\newblock \href {https://doi.org/10.1007/978-3-030-03332-3_15}
  {\path{doi:10.1007/978-3-030-03332-3_15}}.

\bibitem[CLP24]{PKC:CheLerPan24}
Mingjie Chen, Antonin Leroux, and Lorenz Panny.
\newblock {SCALLOP}-{HD}: Group action from 2-dimensional isogenies.
\newblock In Qiang Tang and Vanessa Teague, editors, {\em PKC~2024: 27th
  International Conference on Theory and Practice of Public Key Cryptography,
  Part~III}, volume 14603 of {\em Lecture Notes in Computer Science}, pages
  190--216. Springer, Cham, April 2024.
\newblock \href {https://doi.org/10.1007/978-3-031-57725-3_7}
  {\path{doi:10.1007/978-3-031-57725-3_7}}.

\bibitem[Cou06]{couveignes_hard_2006}
Jean-Marc Couveignes.
\newblock Hard homogeneous spaces.
\newblock {\em Cryptology ePrint Archive}, page 291, 2006.
\newblock URL: \url{https://eprint.iacr.org/2006/291}.

\bibitem[CPV20]{EC:CasPanVer20}
Wouter Castryck, Lorenz Panny, and Frederik Vercauteren.
\newblock Rational isogenies from irrational endomorphisms.
\newblock In Anne Canteaut and Yuval Ishai, editors, {\em Advances in
  Cryptology -- {EUROCRYPT}~2020, Part~II}, volume 12106 of {\em Lecture Notes
  in Computer Science}, pages 523--548. Springer, Cham, May 2020.
\newblock \href {https://doi.org/10.1007/978-3-030-45724-2_18}
  {\path{doi:10.1007/978-3-030-45724-2_18}}.

\bibitem[CS22]{chenu_higher-degree_2022}
Mathilde Chenu and Benjamin Smith.
\newblock Higher-degree supersingular group actions.
\newblock {\em Mathematical Cryptology}, 1(2):85--101, March 2022.
\newblock Publisher: Florida Online Journals.

\bibitem[DDF{\etalchar{+}}21]{AC:DDFKLP21}
Luca {De Feo}, Cyprien {Delpech de Saint Guilhem}, Tako~Boris Fouotsa,
  P{\'e}ter Kutas, Antonin Leroux, Christophe Petit, Javier Silva, and Benjamin
  Wesolowski.
\newblock {S{\'e}ta}: Supersingular encryption from torsion attacks.
\newblock In Mehdi Tibouchi and Huaxiong Wang, editors, {\em Advances in
  Cryptology -- {ASIACRYPT}~2021, Part~IV}, volume 13093 of {\em Lecture Notes
  in Computer Science}, pages 249--278. Springer, Cham, December 2021.
\newblock \href {https://doi.org/10.1007/978-3-030-92068-5_9}
  {\path{doi:10.1007/978-3-030-92068-5_9}}.

\bibitem[DFK{\etalchar{+}}23]{PKC:DFKLMP23}
Luca {De Feo}, Tako~Boris Fouotsa, P{\'e}ter Kutas, Antonin Leroux,
  Simon-Philipp Merz, Lorenz Panny, and Benjamin Wesolowski.
\newblock {SCALLOP}: Scaling the {CSI}-{FiSh}.
\newblock In Alexandra Boldyreva and Vladimir Kolesnikov, editors, {\em
  PKC~2023: 26th International Conference on Theory and Practice of Public Key
  Cryptography, Part~I}, volume 13940 of {\em Lecture Notes in Computer
  Science}, pages 345--375. Springer, Cham, May 2023.
\newblock \href {https://doi.org/10.1007/978-3-031-31368-4_13}
  {\path{doi:10.1007/978-3-031-31368-4_13}}.

\bibitem[DG16]{DCC:DelGal16}
Christina Delfs and Steven~D. Galbraith.
\newblock Computing isogenies between supersingular elliptic curves over
  {$\mathbb{F}_p$}.
\newblock {\em Designs, Codes and Cryptography}, 78(2):425--440, 2016.
\newblock \href {https://doi.org/10.1007/s10623-014-0010-1}
  {\path{doi:10.1007/s10623-014-0010-1}}.

\bibitem[DKL{\etalchar{+}}20]{AC:DKLPW20}
Luca {De Feo}, David Kohel, Antonin Leroux, Christophe Petit, and Benjamin
  Wesolowski.
\newblock {SQISign}: Compact post-quantum signatures from quaternions and
  isogenies.
\newblock In Shiho Moriai and Huaxiong Wang, editors, {\em Advances in
  Cryptology -- {ASIACRYPT}~2020, Part~I}, volume 12491 of {\em Lecture Notes
  in Computer Science}, pages 64--93. Springer, Cham, December 2020.
\newblock \href {https://doi.org/10.1007/978-3-030-64837-4_3}
  {\path{doi:10.1007/978-3-030-64837-4_3}}.

\bibitem[DLRW24]{EC:DLRW24}
Pierrick Dartois, Antonin Leroux, Damien Robert, and Benjamin Wesolowski.
\newblock {SQIsignHD}: New dimensions in cryptography.
\newblock In Marc Joye and Gregor Leander, editors, {\em Advances in Cryptology
  -- {EUROCRYPT}~2024, Part~I}, volume 14651 of {\em Lecture Notes in Computer
  Science}, pages 3--32. Springer, Cham, May 2024.
\newblock \href {https://doi.org/10.1007/978-3-031-58716-0_1}
  {\path{doi:10.1007/978-3-031-58716-0_1}}.

\bibitem[EHL{\etalchar{+}}18]{EC:EHLMP18}
Kirsten Eisentr{\"a}ger, Sean Hallgren, Kristin~E. Lauter, Travis Morrison, and
  Christophe Petit.
\newblock Supersingular isogeny graphs and endomorphism rings: Reductions and
  solutions.
\newblock In Jesper~Buus Nielsen and Vincent Rijmen, editors, {\em Advances in
  Cryptology -- {EUROCRYPT}~2018, Part~III}, volume 10822 of {\em Lecture Notes
  in Computer Science}, pages 329--368. Springer, Cham, April~/~May 2018.
\newblock \href {https://doi.org/10.1007/978-3-319-78372-7_11}
  {\path{doi:10.1007/978-3-319-78372-7_11}}.

\bibitem[EHL{\etalchar{+}}20]{eisentrager_computing_2020}
Kirsten Eisenträger, Sean Hallgren, Chris Leonardi, Travis Morrison, and
  Jennifer Park.
\newblock Computing endomorphism rings of supersingular elliptic curves and
  connections to path-finding in isogeny graphs.
\newblock {\em Open Book Series}, 4(1):215--232, 2020.
\newblock Publisher: Mathematical Sciences Publishers.
\newblock \href {https://doi.org/10.2140/obs.2020.4.215}
  {\path{doi:10.2140/obs.2020.4.215}}.

\bibitem[Kan97]{kani_number_1997}
Ernst Kani.
\newblock The number of curves of genus two with elliptic differentials.
\newblock {\em Journal für die reine und angewandte Mathematik}, 485:93--122,
  1997.
\newblock \href {https://doi.org/10.1515/crll.1997.485.93}
  {\path{doi:10.1515/crll.1997.485.93}}.

\bibitem[Koh96]{kohel_endomorphism_1996}
David~Russell Kohel.
\newblock {\em Endomorphism rings of elliptic curves over finite fields}.
\newblock ProQuest LLC, Ann Arbor, MI, 1996.
\newblock Thesis (Ph.D.)--University of California, Berkeley.

\bibitem[Kup05]{kuperberg_subexponential-time_2005}
Greg Kuperberg.
\newblock A subexponential-time quantum algorithm for the dihedral hidden
  subgroup problem.
\newblock {\em SIAM Journal on Computing}, 35(1):170--188, 2005.
\newblock Publisher: SIAM.
\newblock \href {https://doi.org/10.1137/S0097539703436345}
  {\path{doi:10.1137/S0097539703436345}}.

\bibitem[LR12]{lubicz_computing_2012}
David Lubicz and Damien Robert.
\newblock Computing isogenies between abelian varieties.
\newblock {\em Compositio Mathematica}, 148(5):1483--1515, 2012.
\newblock Publisher: London Mathematical Society.
\newblock \href {https://doi.org/10.1112/S0010437X12000243}
  {\path{doi:10.1112/S0010437X12000243}}.

\bibitem[LR23]{lubicz_fast_2023}
David Lubicz and Damien Robert.
\newblock Fast change of level and applications to isogenies.
\newblock {\em Research in Number Theory}, 9(1):7, 2023.
\newblock Publisher: Springer.
\newblock \href {https://doi.org/10.1007/s40993-022-00407-9}
  {\path{doi:10.1007/s40993-022-00407-9}}.

\bibitem[Mil86]{milne_abelian_1986}
James~S Milne.
\newblock Abelian varieties.
\newblock {\em Arithmetic geometry}, pages 103--150, 1986.
\newblock Publisher: Springer.

\bibitem[MMP{\etalchar{+}}23]{EC:MMPPW23}
Luciano Maino, Chloe Martindale, Lorenz Panny, Giacomo Pope, and Benjamin
  Wesolowski.
\newblock A direct key recovery attack on {SIDH}.
\newblock In Carmit Hazay and Martijn Stam, editors, {\em Advances in
  Cryptology -- {EUROCRYPT}~2023, Part~V}, volume 14008 of {\em Lecture Notes
  in Computer Science}, pages 448--471. Springer, Cham, April 2023.
\newblock \href {https://doi.org/10.1007/978-3-031-30589-4_16}
  {\path{doi:10.1007/978-3-031-30589-4_16}}.

\bibitem[Mum70]{david_mumford_abelian_1970}
David Mumford.
\newblock {\em Abelian {Varieties}}.
\newblock Oxford University Press, London, 1970.

\bibitem[Onu21]{onuki_oriented_2021}
Hiroshi Onuki.
\newblock On oriented supersingular elliptic curves.
\newblock {\em Finite Fields and Their Applications}, 69:101777, 2021.
\newblock Publisher: Elsevier.
\newblock \href {https://doi.org/10.1016/j.ffa.2020.101777}
  {\path{doi:10.1016/j.ffa.2020.101777}}.

\bibitem[Pom87]{pomerance_fast_1987}
Carl Pomerance.
\newblock Fast, rigorous factorization and discrete logarithm algorithms.
\newblock In {\em Discrete algorithms and complexity ({K}yoto, 1986)},
  volume~15 of {\em Perspect. Comput.}, pages 119--143. Academic Press, Boston,
  MA, 1987.

\bibitem[PR23]{page_introducing_2023}
Aurel Page and Damien Robert.
\newblock Introducing {Clapoti} (s): {Evaluating} the isogeny class group
  action in polynomial time.
\newblock {\em Cryptology ePrint Archive}, page 1766, 2023.
\newblock URL: \url{https://eprint.iacr.org/2023/1766}.

\bibitem[PT18]{pollack_finding_2018}
Paul Pollack and Enrique Treviño.
\newblock Finding the {Four} {Squares} in {Lagrange}'s {Theorem}.
\newblock {\em Integers}, 18(A15):7--17, 2018.

\bibitem[PW24]{EC:PagWes24}
Aurel Page and Benjamin Wesolowski.
\newblock The supersingular endomorphism ring and one endomorphism problems are
  equivalent.
\newblock In Marc Joye and Gregor Leander, editors, {\em Advances in Cryptology
  -- {EUROCRYPT}~2024, Part~VI}, volume 14656 of {\em Lecture Notes in Computer
  Science}, pages 388--417. Springer, Cham, May 2024.
\newblock \href {https://doi.org/10.1007/978-3-031-58751-1_14}
  {\path{doi:10.1007/978-3-031-58751-1_14}}.

\bibitem[Rob21]{robert_efficient_2021}
Damien Robert.
\newblock {\em Efficient algorithms for abelian varieties and their moduli
  spaces}.
\newblock {HDR} {Thesis}, Université de Bordeaux (UB), 2021.

\bibitem[Rob22a]{robert_evaluating_2022}
Damien Robert.
\newblock Evaluating isogenies in polylogarithmic time.
\newblock {\em Cryptology ePrint Archive}, page 1068, 2022.
\newblock URL: \url{https://eprint.iacr.org/2022/1068}.

\bibitem[Rob22b]{robert_applications_2022}
Damien Robert.
\newblock Some applications of higher dimensional isogenies to elliptic curves
  (preliminary version).
\newblock {\em Cryptology ePrint Archive}, page 1704, 2022.
\newblock URL: \url{https://eprint.iacr.org/2022/1704}.

\bibitem[Rob23]{EC:Robert23}
Damien Robert.
\newblock Breaking {SIDH} in polynomial time.
\newblock In Carmit Hazay and Martijn Stam, editors, {\em Advances in
  Cryptology -- {EUROCRYPT}~2023, Part~V}, volume 14008 of {\em Lecture Notes
  in Computer Science}, pages 472--503. Springer, Cham, April 2023.
\newblock \href {https://doi.org/10.1007/978-3-031-30589-4_17}
  {\path{doi:10.1007/978-3-031-30589-4_17}}.

\bibitem[Rob24]{robert_efficient_2024}
Damien Robert.
\newblock On the efficient representation of isogenies (a survey).
\newblock {\em Cryptology ePrint Archive}, page 1071, 2024.
\newblock URL: \url{https://eprint.iacr.org/2024/1071}.

\bibitem[RS06]{rostovtsev_public-key_2006}
Alexander Rostovtsev and Anton Stolbunov.
\newblock Public-key cryptosystem based on isogenies.
\newblock {\em Cryptology ePrint Archive}, page 145, 2006.
\newblock URL: \url{https://eprint.iacr.org/2006/145}.

\bibitem[Sil86]{silverman_arithmetic_1986}
Joseph~H. Silverman.
\newblock {\em The arithmetic of elliptic curves}, volume 106 of {\em Graduate
  texts in mathematics}.
\newblock Springer, 1986.
\newblock \href {https://doi.org/10.1007/978-1-4757-1920-8}
  {\path{doi:10.1007/978-1-4757-1920-8}}.

\bibitem[Voi21]{voight_quaternion_2021}
John Voight.
\newblock {\em Quaternion algebras}.
\newblock Springer Nature, 2021.

\bibitem[Vé71]{velu_isogenies_1971}
Jacques Vélu.
\newblock Isogénies entre courbes elliptiques.
\newblock {\em C. R. Acad. Sc. Paris, Série A}, t. 273:238--241, 1971.

\bibitem[Wes22a]{EC:Wesolowski22}
Benjamin Wesolowski.
\newblock Orientations and the supersingular endomorphism ring problem.
\newblock In Orr Dunkelman and Stefan Dziembowski, editors, {\em Advances in
  Cryptology -- {EUROCRYPT}~2022, Part~III}, volume 13277 of {\em Lecture Notes
  in Computer Science}, pages 345--371. Springer, Cham, May~/~June 2022.
\newblock \href {https://doi.org/10.1007/978-3-031-07082-2_13}
  {\path{doi:10.1007/978-3-031-07082-2_13}}.

\bibitem[Wes22b]{wesolowski_supersingular_2022}
Benjamin Wesolowski.
\newblock The supersingular isogeny path and endomorphism ring problems are
  equivalent.
\newblock In {\em 2021 {IEEE} 62nd {Annual} {Symposium} on {Foundations} of
  {Computer} {Science} ({FOCS})}, pages 1100--1111. IEEE, 2022.
\newblock \href {https://doi.org/10.1109/FOCS52979.2021.00109}
  {\path{doi:10.1109/FOCS52979.2021.00109}}.

\bibitem[Wes24]{wesolowski_random_2024}
Benjamin Wesolowski.
\newblock {\em Random {Walks} in {Number}-theoretic {Cryptology}}.
\newblock {HDR} {Thesis}, ENS Lyon, 2024.
\newblock URL: \url{https://bweso.com/hdr.pdf}.

\bibitem[Zar74]{zarhin_remark_1974}
Ju~G Zarhin.
\newblock A remark on endomorphisms of abelian varieties over function fields
  of finite characteristic.
\newblock {\em Mathematics of the USSR-Izvestiya}, 8(3):477, 1974.
\newblock Publisher: IOP Publishing.
\newblock \href {https://doi.org/10.1070/IM1974v008n03ABEH002115}
  {\path{doi:10.1070/IM1974v008n03ABEH002115}}.

\end{thebibliography}
\end{document}